\documentclass[11pt]{article}
\usepackage{fix-cm}
\usepackage[dvipsnames]{xcolor}
\usepackage[letterpaper,margin=1.25in,marginparwidth=1in,marginparsep=0.1in]{geometry}
\usepackage{microtype}
\usepackage{mathtools,amssymb,amsthm,thmtools}
\usepackage{enumitem}
\setlist[enumerate]{label=(\roman*),ref=(\roman*)}
\usepackage{needspace}
\usepackage{newpxmath,newpxtext}
\let\epsilon\varepsilon %
\let\emptyset\varnothing %
\usepackage{hyperref}
\usepackage{tikz}
\usetikzlibrary{decorations.pathreplacing}

\definecolor{definitioncolor}{HTML}{B8342E}
\definecolor{cyc0}{HTML}{E3E3E3}
\definecolor{cyc1}{HTML}{A6A6A6}
\definecolor{cyc2}{HTML}{555555}
\definecolor{mylinkcolor}{HTML}{2B5B9C}
\colorlet{mycitecolor}{mylinkcolor}
\definecolor{myurlcolor}{HTML}{1B7A72}
\hypersetup{
  colorlinks=true,
  linkcolor=mylinkcolor,
  citecolor=mycitecolor,
  urlcolor=myurlcolor,
  pdftitle={Settling the Approximability of Submodular Facility Location}
}

\usepackage[paper,no patch,notion]{knowledge}
\knowledgestyle*{notion}{color=black}
\knowledgestyle*{intro notion}{emphasize,color=black}
\usepackage[ruled,vlined,linesnumbered]{algorithm2e}
\SetKwFor{For}{for}{}{}
\SetKwIF{If}{ElseIf}{Else}{if}{}{else if}{else}{}
\SetArgSty{textup}

\usepackage[capitalise,noabbrev]{cleveref}
\crefname{algocf}{Algorithm}{Algorithms}
\crefname{AlgoLine}{line}{lines}
\Crefname{AlgoLine}{Line}{Lines}
\crefname{fact}{Fact}{Facts}
\crefname{question}{Question}{Questions}
\crefname{enumi}{part}{parts}
\Crefname{enumi}{Part}{Parts}
\NewDocumentCommand\itemref{smm}{%
  \IfBooleanF{#1}{\nameCref{#2}~}\hyperref[#3]{\ref*{#2}\,\ref*{#3}}}
\newtheorem{theorem}{Theorem}
\newtheorem{lemma}[theorem]{Lemma}
\newtheorem{corollary}[theorem]{Corollary}
\newtheorem{proposition}[theorem]{Proposition}

\theoremstyle{definition}
\newtheorem{definition}[theorem]{Definition}
\newtheorem{remark}[theorem]{Remark}

\theoremstyle{plain}
\newcommand\restatedlemmanumber{}
\newtheorem*{restatedlemmainner}{Lemma~\restatedlemmanumber}
\NewDocumentEnvironment{restatedlemma}{mo}
  {\renewcommand\restatedlemmanumber{\ref{#1}}%
   \IfValueTF{#2}{\begin{restatedlemmainner}[#2]}{\begin{restatedlemmainner}}}
  {\end{restatedlemmainner}}

\knowledgenewrobustcmd\dist{\mathop{\cmdkl{\mathrm{dist}}}\nolimits}
\knowledgenewrobustcmd\OPT{\mathord{\cmdkl{\mathrm{OPT}}}}
\knowledgenewrobustcmd\Cost{\mathop{\cmdkl{\mathrm{cost}}}\nolimits}
\knowledgenewrobustcmd\BipGraph{\mathord{\cmdkl{G}}}
\knowledgenewrobustcmd\Nbr{\mathord{\cmdkl{N_G}}}
\knowledgenewrobustcmd\LiftNbr{\mathord{\cmdkl{N_{\widetilde G}}}}
\knowledgenewrobustcmd\CoverNum{\mathord{\cmdkl\tau}}
\knowledgenewrobustcmd\AbstractInputCost{\mathord{\cmdkl{g}}}
\knowledgenewrobustcmd\AbstractOrbit{\mathord{\cmdkl{P}}}
\knowledgenewrobustcmd\AbstractOrbitIndex{\mathord{\cmdkl{\omega}}}
\knowledgenewrobustcmd\AbstractGroup{\mathord{\cmdkl{\Gamma}}}
\knowledgenewrobustcmd\AbstractLabelMap{\mathord{\cmdkl{\pi}}}
\knowledgenewrobustcmd\AbstractCopies{\mathord{\cmdkl{R}}}
\knowledgenewrobustcmd\AbstractPenaltyWeight{\mathord{\cmdkl{\lambda}}}
\knowledgenewrobustcmd\AbstractCopyCount{\mathord{\cmdkl{\mu}}}
\knowledgenewrobustcmd\AbstractBlockAvg{\mathord{\cmdkl{\bar\mu}}}
\knowledgenewrobustcmd\AbstractImbalance{\mathord{\cmdkl{\xi}}}
\knowledgenewrobustcmd\AbstractHidCost{\mathord{\cmdkl{g_\pi}}}
\knowledgenewrobustcmd\AbstractRefCost{\mathord{\cmdkl{g_{\mathrm{ref}}}}}
\knowledgenewrobustcmd\AbstractInputOpt{\mathord{\cmdkl{\mathrm{OPT}_0}}}
\knowledgenewrobustcmd\AbstractSymOpt{\mathord{\cmdkl{\mathrm{OPT}_{\mathrm{sym}}}}}
\knowledgenewrobustcmd\AbstractHidOpt{\mathord{\cmdkl{\mathrm{OPT}_\pi}}}
\knowledgenewrobustcmd\AbstractRefOpt{\mathord{\cmdkl{\mathrm{OPT}_{\mathrm{ref}}}}}
\knowledgenewrobustcmd\AbstractQuotCost{\mathord{\cmdkl{g_{\mathrm{quot}}}}}
\knowledgenewrobustcmd\AbstractDist{\mathop{\cmdkl{\mathrm{dist}}}\nolimits}
\knowledgenewrobustcmd\AbstractBlock{\mathord{\cmdkl{B}}}
\knowledgenewrobustcmd\AbstractSplitPen{\mathord{\cmdkl{\psi}}}
\knowledgenewrobustcmd\SparseCost{\mathord{\cmdkl{\rho}}}
\knowledgenewrobustcmd\SparseObj{\mathord{\cmdkl{\Phi}}}
\knowledgenewrobustcmd\DiagSet{\mathord{\cmdkl{D}}}
\knowledgenewrobustcmd\Transv{\mathord{\cmdkl{\mathcal T}}}
\knowledgenewrobustcmd\Div{\mathord{\cmdkl{D}}}
\knowledgenewrobustcmd\JrpLiftCost{\mathord{\cmdkl{g}}}
\knowledgenewrobustcmd\JrpBaseCost{\mathord{\cmdkl{g_0}}}
\knowledgenewrobustcmd\JrpBaseOpt{\mathord{\cmdkl{\mathrm{OPT}_{\mathrm{base}}}}}
\knowledgenewrobustcmd\ItemFiber{\mathord{\cmdkl{P}}}
\knowledgenewrobustcmd\ServiceCost{\mathord{\cmdkl{g}}}
\knowledgenewrobustcmd\JrpCost{\mathord{\cmdkl{g}}}
\knowledgenewrobustcmd\Code{\mathord{\cmdkl{\mathsf{C}}}}
\knowledgenewrobustcmd\Enc{\mathord{\cmdkl{\mathsf{Enc}}}}
\knowledgenewrobustcmd\Dec{\mathord{\cmdkl{\mathsf{Dec}}}}
\knowledgenewrobustcmd\ListDec{\mathord{\cmdkl{\mathsf{ListDec}}}}
\knowledgenewrobustcmd\SymbolMap{\mathord{\cmdkl{\iota}}}
\makeatletter
\newcommand\AlignedIntro[1]{\ifmeasuring@#1\else\phantomsection\intro*#1\fi}
\makeatother
\knowledge{ground set}{notion}
\knowledge{assignment}{notion}
\knowledge{submodular allocation}{notion}
\knowledge{normalized}{notion}
\knowledge{cluster}{notion}
\knowledge{homogeneous}{notion}
\knowledge{value oracle}{notion}
\knowledge{adaptive}{notion}
\knowledge{expected-cost}{notion}
\knowledge{success-probability}{notion}
\knowledge{polynomial query bound}{notion}
\knowledge{bipartite graph}{notion}
\knowledge{neighborhood}{notion}
\knowledge{cover number}{notion}
\knowledge{degree}{notion}
\knowledge{d-regular}{notion}
\knowledge{witness}{notion}
\knowledge{minimizing witness}{notion}
\knowledge{cap}{notion}
\knowledge{triangle inequality}{notion}
\knowledge{fiber}{notion}
\knowledge{base graph}{notion}
\knowledge{cyclic lift}{notion}
\knowledge{lifted client}{notion}
\knowledge{lifted right vertex}{notion}
\knowledge{cyclic coordinate}{notion}
\knowledge{whole fibers}{notion}
\knowledge{lift}{notion}
\knowledge{interval submodular allocation}{notion}
\knowledge{window}{notion}
\knowledge{matching}{notion}
\knowledge{block family}{notion}
\knowledge{transversal}{notion}
\knowledge{selector}{notion}
\knowledge{diversity}{notion}
\knowledge{profile}{notion}
\knowledge{interval selector}{notion}
\knowledge{weighted diversity}{notion}
\knowledge{weighted interval-diversity instance}{notion}
\knowledge{permutation group}{notion}
\knowledge{invariant}{notion}
\knowledge{shift}{notion}
\knowledge{lifted instance}{notion}
\knowledge{symmetry gap}{notion}
\knowledge{matroid}{notion}
\knowledge{Lovasz extension}{notion}
\knowledge{rank function}{notion}
\knowledge{multilinear extension}{notion}
\knowledge{instance}{notion}
\knowledge{abstract symmetric instance}{notion}
\knowledge{abstract symmetry group}{notion}
\knowledge{abstract orbit}{notion}
\knowledge{abstract input facility cost function}{notion}
\knowledge{abstract input optimum}{notion}
\knowledge{abstract symmetric optimum}{notion}
\knowledge{abstract hidden optimum}{notion}
\knowledge{abstract reference optimum}{notion}
\knowledge{abstract hidden instance}{notion}
\knowledge{abstract reference instance}{notion}
\knowledge{abstract replication factor}{notion}
\knowledge{abstract penalty weight}{notion}
\knowledge{abstract threshold sets}{notion}
\knowledge{abstract quotient instance}{notion}
\knowledge{abstract quotient function}{notion}
\knowledge{abstract symmetric fractional assignment}{notion}
\knowledge{configuration}{notion}
\knowledge{configuration LP}{notion}
\knowledge{whole-orbit}{notion}
\knowledge{abstract replication}{notion}
\knowledge{abstract public block}{notion}
\knowledge{abstract copy class}{notion}
\knowledge{abstract label}{notion}
\knowledge{abstract copy}{notion}
\knowledge{abstract hidden partition}{notion}
\knowledge{abstract copy count}{notion}
\knowledge{abstract block average}{notion}
\knowledge{abstract imbalance}{notion}
\knowledge{abstract well balanced}{notion}
\knowledge{abstract split penalty}{notion}
\knowledge{abstract hidden function}{notion}
\knowledge{abstract reference function}{notion}
\knowledge{admissible}{notion}
\knowledge{local expansion}{notion}
\knowledge{sparse rank function}{notion}
\knowledge{first copy neighborhood}{notion}
\knowledge{first copy assignment}{notion}
\knowledge{linear code}{notion}
\knowledge{block length}{notion}
\knowledge{polynomially constructible}{notion}
\knowledge{encoder}{notion}
\knowledge{decoder}{notion}
\knowledge{list decoder}{notion}
\title{Settling the Approximability of\\
Submodular Facility Location}
\author{%
  Adam N. Elmachtoub\thanks{Department of Industrial Engineering and
  Operations Research, Columbia University, New York, NY, USA.}
  \and Bento Natura\footnotemark[1]
  \and Meike Neuwohner\thanks{CNRS \& Department of Computer Science
  (DI ENS), \'Ecole Normale Sup\'erieure, Paris, France.}}
\date{}
\begin{document}
\maketitle

\begin{abstract}
In submodular facility location (SFL), introduced by Svitkina and
Tardos (SODA 2006), clients $\cal C$ in a metric space are assigned to
facilities $\cal F$. A facility serving a set $S \subseteq \cal C$ of
clients incurs cost $g(S)$, where $g$ is a monotone submodular function
common to all facilities. The best known approximation factor for SFL
is $O(\log\log|\cal C|)$, due to Abbasi, Adamczyk, Bosch-Calvo, Byrka,
Grandoni, Sornat, and Tinguely (ICALP 2024), improving
on the $O(\log|\cal C|)$-approximation of Svitkina and Tardos. We show that this
is best possible up to a constant factor in the value-oracle model by showing that there exists no algorithm with only polynomially many queries to $g$ with approximation ratio $o(\log\log|\cal C|)$. This bound holds for algorithms with unlimited computation
time, and when $g$ is restricted to matroid rank functions.
We also show that the natural linear programming relaxation has an
integrality gap of $\Theta(\log\log|\cal C|)$. Our lower bound follows by constructing an extension of Vondr\'ak's symmetry-gap approach (FOCS 2009) for matroid rank functions. 
\end{abstract}
 
\newpage

\section{Introduction}\label{sec:strategy}\label{sec:model}
In metric facility location, clients in a metric space are assigned to
facilities; every client pays the distance to its facility, and every
facility pays a cost for serving its clients. Svitkina and
Tardos~\cite{ST10} introduced \emph{submodular facility
location} (SFL), in which the cost of a facility is a monotone
submodular function of the set of clients it serves. Submodularity models economies of
scale in serving clients together. We consider the homogeneous version,
in which all facilities share one \emph{facility cost function}, following
the terminology of Svitkina and Tardos; the authors of~\cite{Abbasi2024}
call these \emph{opening costs}. Thus, given clients $\mathcal C$ and
facilities $\mathcal F$ in a metric space and a monotone submodular
function $\ServiceCost\colon2^{\mathcal C}\to\mathbb R_{\geq0}$, the goal
is to assign every client to a facility so as to minimize the sum, over
all facilities $f$, of $\ServiceCost(C_f)$ and the distances from the
clients in $C_f$ to $f$, where $C_f$ is the set of clients assigned
to~$f$. We define SFL formally in \cref{sec:config-preliminaries}, as a
special case of an allocation problem with general connection costs. 

If the facility cost function takes the same positive value on every nonempty
set, SFL is metric uncapacitated facility location with uniform opening
costs. This special case is APX-hard: Guha and
Khuller~\cite{GuhaKhuller1999} showed that it cannot be
approximated in polynomial time within a factor $1.463$ unless
$\mathrm{NP}\subseteq\mathrm{DTIME}(n^{O(\log\log n)})$.
If the facilities have different cost functions, then the problem contains set cover, for which the factor $O(\log|\mathcal C|)$
is best possible unless
$\mathrm{P}=\mathrm{NP}$~\cite{Feige1998,DinurSteurer2014}.
Svitkina and Tardos \cite{ST10} gave an $O(\log|\mathcal C|)$ approximation, even
when every facility has its own submodular facility cost function, and a
constant factor approximation when all facilities share a facility cost function
given by a node-weighted tree whose leaves are the clients. They asked whether SFL admits a constant-factor
approximation. Abbasi, Adamczyk, Bosch-Calvo, Byrka, Grandoni, Sornat,
and Tinguely~\cite{Abbasi2024} improved the factor to
$O(\log\log|\mathcal C|)$ by rounding the configuration LP. Their
algorithm accesses the facility cost function through a value oracle,
the standard model for general submodular functions, and solves the
configuration LP using submodular minimization. We define this model
next.

\begin{definition}[Value-oracle model]\label{def:oracle}
The client and facility sets and their connection costs are public, i.e., known to the algorithm.
An exact \intro[value oracle]{value oracle} returns $\ServiceCost(C)$ for any queried $C\subseteq\mathcal C$.
Queries may be \intro[adaptive]{adaptive} (i.e., depending on previous answers)
and computation between queries is unrestricted. A
\intro[polynomial query bound]{polynomial query bound} means that, for
some polynomial $Q$, the algorithm makes at most
$Q(|\mathcal C|+|\mathcal F|)$ queries on every instance, regardless of
its random choices.
An algorithm achieves factor $\alpha$ with
\intro[success-probability]{success probability} $p$ if, on every
instance, it returns an assignment of cost at most $\alpha\,\OPT$ with
probability at least $p$ over its random choices.
\end{definition}

\subsection{Our results}

Our main result is a matching lower bound: in the value-oracle model, no algorithm with polynomially many queries achieves factor $o(\log\log|\mathcal C|)$. Hence the $O(\log\log|\mathcal C|)$-approximation of~\cite{Abbasi2024} is asymptotically optimal.

\begin{theorem}[Lower bound for SFL]\label{thm:quantitative}
There is a constant $\alpha_0>0$ such that no algorithm with a \kl{polynomial query
bound} achieves factor $\alpha_0\log\log|\mathcal C|$ with constant
\kl[success-probability]{success probability} on SFL instances with
matroid-rank facility cost functions.
\end{theorem}

In particular, no algorithm with a
\kl{polynomial query bound} achieves factor $o(\log\log|\mathcal C|)$
with constant \kl[success-probability]{success probability}, and
SFL admits no constant-factor approximation in the
value-oracle model. 
The proof of \cref{thm:quantitative} actually gives a slightly stronger statement, namely a quasi-polynomial lower bound
on the number of queries required to achieve factor $\alpha_0\log\log|\mathcal C|$ with a 
\kl[success-probability]{success probability} whose reciprocal is polynomially bounded in the number of clients (\cref{rem:quasi-polynomial}). The same
instances considered in the proof also show that the
configuration LP has integrality gap $\Omega(\log\log|\mathcal C|)$
(\cref{thm:sfl-config-gap}), so the LP-based analysis
of~\cite{Abbasi2024} is tight as well.

The lower bound follows from a general \emph{symmetry-gap theorem}
(\cref{thm:symmetry-gap}),
in the spirit of Vondr\'ak's symmetry-gap method~\cite{Vondrak2013},
which we outline in \cref{sec:overview} and develop in
\cref{sec:symmetry-gap}. The theorem applies to submodular allocation with
arbitrary connection costs and may be of independent interest.

Using the technique of Dobzinski and Vondr\'ak~\cite{DobzinskiVondrak2012},
our information-theoretic lower bound can be converted into computational
hardness, see \cref{app:computational}.

\subsection{Overview of techniques}
\label{sec:overview}

We now give a high-level overview of the proofs. 
We follow the strategy of Vondr\'ak's symmetry-gap
method~\cite{Vondrak2013}. We start from an instance together with a
group of permutations of the clients such that all permutations in the group preserve the facility cost
function and the connection costs. An assignment is
\emph{whole-orbit} if it sends each orbit of the group to a single
facility. We will require instances with a large \emph{symmetry gap}, in which
some assignment that splits orbits is much cheaper than every
whole-orbit assignment. Such an instance on its own does not give a
lower bound, as an algorithm may simply find the cheap assignment. The
idea is to hide how the orbits have to be split, so that an algorithm
with few queries cannot do better than a whole-orbit assignment. We
explain in \cref{sec:related-work} why we do not use Vondr\'ak's
construction directly.

To create the instance for our symmetry-gap theorem, we \emph{replicate} each client
many times. The orbit of a copy
is public, i.e., known to the algorithm, however precise correspondence between original clients and copied clients within an orbit is determined by a hidden random partition. 
We define a
new submodular facility cost function on the copies with a large symmetry gap and such that, with
high probability, an algorithm with few queries cannot obtain any
information about which client each copy was created from, and in particular, it cannot recover a cheap way to split the orbits. More precisely, we compare the
new function with a public reference function, which depends on a set of copies only through the number of
copies it contains from each orbit. The two functions agree on every
set that contains about equally many copies of each client of an
orbit, and by concentration, every fixed set has this property with
high probability over the random hidden information. An algorithm run
against the reference function chooses its queries without the hidden
information. Hence, with high probability, an algorithm with few
queries receives the same answers from both functions, and its output
has the same cost under both.
Under the reference function, every assignment costs at least the
number of copies per client times the cost of the best whole-orbit
assignment. The reference instance also gives configuration-LP gaps:
its integral optimum is large, but symmetrizing a cheap assignment
gives a cheap fractional solution.

It remains to construct an instance of SFL with a large symmetry gap
(\cref{sec:sfl-lower-bound}). We start from a random $d$-regular
bipartite graph with $m \gg d$ vertices on each side. Its right vertices are
the facilities, while copies of the left vertices will become the clients. One can think of the distances as the shortest paths in the graph. A random graph has a
large cover number: covering all left vertices with neighborhoods of
right vertices requires about $(m\log d)/d$ right vertices. Note that this is a factor of $\log d$ more than the lower bound of $m/d$ implied by $d$-regularity.

To turn this into a symmetry gap, we lift every vertex $v$ to a set of $d$ copies $((v,t))_{t\in \mathbb{Z}_d}$
and replace every edge by a perfect matching between the lifts of 
its left and its right endpoint. This replacement is done in a way that ensures the following two properties:
\begin{itemize}
    \item Every lifted vertex $(v,t)$ has exactly one neighbor per copy coordinate $t'$, i.e., one neighbor of the form $(v',t')$ for each $t'\in \mathbb{Z}_d$.
    \item Rotating all copy coordinates, i.e., mapping $(v,t)$ to $(v,t+a)$ yields a graph isomorphism for every $a\in \mathbb{Z}_d$.
\end{itemize}
In particular, the first property ensures that the neighborhoods of the right lifts of the form $(v,0)$ partition the clients.

The properties above can be achieved via a standard
construction, known as a cyclic lift or cyclic cover of the graph, and
in topological graph theory as the derived graph of a voltage
graph~\cite{Gross1974,AmitLinial2002,AgarwalEtAl2019}. The lifts only serve to define the facility
cost function. We would like to choose the facility cost function in such a way that it corresponds to the minimum number of lifted right vertices whose neighborhoods cover a given subset of left clients. By the cyclic construction, sending the neighborhood of
the first lift of each facility to that facility then costs $1$ per lift, so $m$ in total. This assignment splits every orbit. Covering a whole orbit instead requires $d$
lifts, intuitively all $d$ lifts of a neighboring facility.
A whole-orbit assignment therefore pays for $d$ covers of the graph,
i.e., about $d \cdot (m \log d / d) = m\log d$. This gives a gap of
order $\log d$.

This cover function, however, is not submodular. We replace it by a matroid
rank function that agrees with it up to a factor of $3$ on sets of small value.
For this we use that small sets of right vertices in a random $d$-regular graph
have almost $d$ times as many neighbors, so the neighborhoods of few
lifts hardly overlap. We require that this expansion holds for sets of size
$d^{\Theta(d)}$, so the graph has $m = d^{\Theta(d)}$ vertices. With $|\mathcal C| = dm = d^{\Theta(d)}$, the gap of $\log d$ translates to the order of $\log\log|\mathcal C|$.

\subsection{Related work}\label{sec:related-work}

It is a common approach in value-oracle lower bounds to hide structure in the queried function.
One constructs a hidden function and a public reference function that
agree on every query that is not correlated with the hidden structure.
Feige, Mirrokni, and Vondr\'ak~\cite{FeigeMirrokniVondrak2011} used this idea
for submodular maximization, Svitkina and
Fleischer~\cite{SvitkinaFleischer2011} for submodular minimization
problems such as load balancing, sparsest cut, and minimization under
a cardinality lower bound, and
Goemans, Harvey, Iwata, and Mirrokni~\cite{GoemansEtAl2009} for
approximating submodular functions everywhere, i.e., constructing from polynomially many value queries a function that approximates $g$ within a fixed factor on every set. In particular, their result applies to matroid rank functions.

Vondr\'ak's symmetry-gap method~\cite{Vondrak2013} derives such lower
bounds from symmetric instances by refinement and by smoothing the
multilinear extension. Ene, Vondr\'ak, and Wu~\cite{EneVondrakWu2013}
extended the method to minimization over partitions and related
symmetry gaps to LP gaps.

Our symmetric instance also fits into Vondr\'ak's framework, in the
minimization form of Ene, Vondr\'ak, and Wu. We nevertheless prove our 
own symmetry-gap theorem, for two reasons. First, in Vondr\'ak's construction the
hidden and reference functions are only guaranteed to agree on queries
that contain almost exactly the same fraction of the copies of every
client in an orbit, up to an error that is exponentially small in the
input size (for a fixed relative tolerance). Hiding the clients then
requires exponentially many copies per client, which loses another
logarithm in the lower bound.
Our replication only needs polynomially many copies. Second, our theorem
produces matroid-rank facility cost functions. 

\subsection{Organization}
The rest of the paper is organized as follows. We give preliminaries
in \cref{sec:definitions}. In \cref{sec:symmetry-gap},
we state the symmetry-gap theorem, together with its configuration-LP version. \cref{sec:sfl-lower-bound} constructs
the SFL input and applies the symmetry-gap theorem. The same section also proves the integrality gap of the configuration LP.

\section{Preliminaries}\label{sec:definitions}
In this section, we discuss basic results about submodular functions and 
symmetry and introduce submodular allocation, SFL, and the
configuration LP.
For a finite set $E$, $2^E$ denotes its power set; $\sqcup$ denotes disjoint unions, and $[q]\coloneqq\{1,\ldots,q\}$. Throughout, $\log$ denotes the natural logarithm.

\subsection{Submodular functions and matroid rank functions}
\label{sec:statement}\label{sec:shared-tools}
A function $r\colon2^E\to\mathbb R$ on a finite \intro{ground set} $E$
is normalized if $r(\varnothing)=0$, monotone if $r(A)\leq r(B)$ for
$A\subseteq B$, and submodular if
\[
 r(A)+r(B)\geq r(A\cup B)+r(A\cap B)\qquad(A,B\subseteq E).
\]
Adding or subtracting a modular function $A\mapsto w_0+\sum_{e\in A}w_e$
for $w \in \mathbb{R}^E$ preserves submodularity. A normalized nonnegative submodular function
is subadditive: $r(A\cup B)\leq r(A)+r(B)$.

Equivalently, submodularity means diminishing returns:
$r(A\cup\{e\})-r(A)\geq r(B\cup\{e\})-r(B)$ whenever
$A\subseteq B$ and $e\notin B$.
\begin{lemma}\label{fact:concave-cardinality}
Given $W\subseteq E$ and a concave function $g\colon \mathbb{R}\rightarrow\mathbb{R}$, the function $f \colon 2^E\rightarrow \mathbb{R},A\mapsto g(|A\cap W|)$ is submodular. If $g$ is moreover nondecreasing, then $f$ is monotone.
\end{lemma}

It is a well-known fact that every integer-valued, normalized, monotone submodular function $r:2^E\rightarrow\mathbb{Z}_{\ge 0}$ with
$r(A)\leq|A|$ for all $A\subseteq E$ is the \intro{rank function} of a \intro{matroid} on the ground set $E$.

The \intro[Lovasz extension]{Lov\'asz extension} of $r\colon2^E\to\mathbb R$ is the
function $\hat r\colon[0,1]^E\to\mathbb R$ given by
\[
 \hat r(x)\coloneqq\mathbb E_t\,r(\{e\in E:x_e\geq t\}),
 \qquad t\text{ uniform in }(0,1].
\]
It extends $r$, i.e., $\hat r(\mathbf 1_A)=r(A)$ for all $A\subseteq E$,
and it is convex if and only if $r$ is submodular~\cite{Lovasz1983}.
If $r$ is submodular, then by convexity and Jensen's inequality,
$\hat r(x)\leq\mathbb E\,r(A)$ for every random set $A\subseteq E$ with
$\mathbb E\,\mathbf 1_A=x$. Since the random threshold sets $\{e:x_e\geq t\}$ have marginals $x$, $\hat r(x)$ is the minimum of $\mathbb E\,r(A)$ over
all such distributions.

\subsection{Submodular allocation and the configuration LP}
\label{sec:config-preliminaries}

We introduce SFL as a special case of the following allocation problem
with general connection costs. Our symmetry-gap theorem
(\cref{sec:symmetry-gap}) applies to this more general version.

\begin{definition}[Submodular allocation]\label{def:allocation}
An \intro[instance]{instance} $(\mathcal C,\mathcal F,\dist,\ServiceCost)$ of \intro[homogeneous]{homogeneous}
\intro{submodular allocation} consists of disjoint
nonempty finite sets of clients $\mathcal C$ and facilities
$\mathcal F$, connection costs
$\intro*\dist\colon\mathcal C\times\mathcal F\to[0,\infty]$, and monotone submodular facility cost function
$\intro*\ServiceCost\colon2^{\mathcal C}\to\mathbb R_{\geq0}$, which is \intro{normalized}. An
\intro[assignment]{assignment} $\sigma\colon\mathcal C\to\mathcal F$ is
feasible if $\dist(c,\sigma(c))<\infty$ for every client $c$. It
partitions $\mathcal C$ into \intro[cluster]{clusters}
$C_f=\sigma^{-1}(f)$, with cost
\[
 \intro*\Cost(\sigma)
 \coloneqq\sum_{f\in \mathcal F}\left(\ServiceCost(C_f)+\sum_{c\in C_f}\dist(c,f)\right).
\]
Write $\intro*\OPT$ for the minimum cost of a feasible assignment.
\end{definition}

This is a special case of the minimum submodular-cost allocation
problem of Chekuri and Ene~\cite{ChekuriEne2011}, with the submodular
cost $\ServiceCost(S)+\sum_{c\in S}\dist(c,f)$ at facility $f$.

\begin{definition}[Submodular facility location]\label{def:sfl}
\emph{Submodular facility location} (SFL) is submodular allocation
restricted to instances in which $\dist$ is the restriction of a metric
on $\mathcal C\cup\mathcal F$ to $\mathcal C\times\mathcal F$, in particular, all connection costs are finite.
\end{definition}
In SFL, distinct clients may have distance zero, i.e., several
clients may share a location.

A \intro{configuration} at a facility $f$ is a possibly empty set
$C\subseteq\mathcal C$ of clients.
The \intro{configuration LP} has a variable $y_{f,C}\geq0$ for every
facility $f$ and configuration $C$ at $f$ and the following objective and constraints.
\begin{equation}\label{eq:abstract-config-lp}
\begin{aligned}
 \operatorname{ConfLP}\coloneqq\min\quad&
 \sum_{f\in \mathcal F}\sum_{C \subseteq \mathcal C} y_{f,C}
   \left(g(C)+\sum_{c\in C}\dist(c,f)\right)\\
 \text{subject to}\quad&
 \sum_{f\in \mathcal F}\sum_{C\ni c}y_{f,C}=1
 &&(c\in\mathcal C),\\
 &\sum_{C\subseteq \mathcal C} y_{f,C}=1 &&(f\in \mathcal F),\\
 &y_{f,C}\geq0 &&(f\in \mathcal F,\ C\subseteq\mathcal C).
\end{aligned}
\end{equation}
Its integral solutions are exactly the feasible assignments.

\subsection{Minimization and symmetry}
\label{sec:symmetry-preliminaries}

If a function is defined as a minimum over some set, we call the
elements of this set \intro[witness]{witnesses}, and a witness that
attains the minimum a \intro{minimizing witness}. We use the following
facts about minima of submodular functions.
\begin{lemma}[Minimization facts]\label{lem:shared-minimization}
For finite ground sets:
\begin{enumerate}
\item\label{item:min-partial} If $H\colon2^{E\sqcup W}\to\mathbb R$ is submodular, then
$h\colon2^E\to\mathbb R$, $h(S)\coloneqq\min_{A\subseteq W}H(S\cup A)$, is submodular.
\item\label{item:min-rank} Let $\varphi\colon E\to W$ be a map, and let $r\colon2^W\to\mathbb R_{\geq0}$
be submodular and \kl{normalized}. Then, the function
\[
h: 2^E\rightarrow \mathbb{R}_{\ge 0}, \quad h(S)\coloneqq\min_{A\subseteq W}\bigl(|S\setminus\varphi^{-1}(A)|+r(A)\bigr)
\]
is normalized, monotone, and submodular, with $h(S)\leq|S|$. If $r$ is
integer-valued, $h$ is a matroid rank function.
\end{enumerate}
\end{lemma}
\begin{proof}
For \cref{item:min-partial}, pick minimizing witnesses $A,B\subseteq W$ for $S,T\subseteq E$. By submodularity of $H$, we have 
\begin{align*}
H((S\cap T)\cup (A\cap B))+H((S\cup T)\cup (A\cup B)) &= H((S\cup A)\cap (T\cup B))+H((S\cup A)\cup (T\cup B)) \\
& \le H(S\cup A)+H(T\cup B) \,.
\end{align*}
Moreover, $A\cap B$ and $A\cup B$ are witnesses for $S\cap T$ and $S\cup T$.
For \cref{item:min-rank}, note that
$|S\setminus\varphi^{-1}(A)|=|S\cup\varphi^{-1}(A)|-|\varphi^{-1}(A)|$.
The function $S\sqcup A\mapsto |S\cup\varphi^{-1}(A)|$ is a coverage function, in which
each element of $E$ covers itself and each $w\in W$ covers
$\varphi^{-1}(w)$. Coverage functions are submodular. Moreover, the function $S\sqcup A\mapsto |\varphi^{-1}(A)|$ is modular and subtracting it retains submodularity. The minimization objective is
therefore submodular on $E\sqcup W$, and \cref{item:min-partial} applies. It is monotone in
$S$ for each $A$, the witness $A=\varnothing$ gives $h(S)\leq|S|$, and
$h(\varnothing)=\min_Ar(A)=0$. If $r$ is integer-valued, so is $h$.
Using the above characterization, $h$ is the rank function of a matroid.
\end{proof}

A \intro{permutation group} on a finite set $E$ is a set $\Gamma$ of
bijections $E\to E$ that contains the identity and is closed under
composition and inverses.
Let $\Gamma$ be a permutation group on $E$. For
$\gamma\in\Gamma$ and $A\subseteq E$, the translate of $A$ by $\gamma$
is $\gamma(A)\coloneqq\{\gamma(e):e\in A\}$, and the orbit of $e\in E$ is
$\Gamma(e)\coloneqq\{\gamma(e):\gamma\in\Gamma\}$. The orbits partition $E$.
A set $A\subseteq E$ is \intro{invariant} under $\Gamma$ if $\gamma(A)=A$ for
all $\gamma\in\Gamma$, equivalently if it is a union of orbits.
A set function $r\colon2^E\to\mathbb R$ is invariant under $\Gamma$ if
$r(\gamma(A))=r(A)$ for all $\gamma\in\Gamma$ and $A\subseteq E$.

We use the following facts about submodular functions that are invariant under a permutation group.
\begin{lemma}[Symmetry facts]\label{lem:shared-symmetry}
For finite ground sets:
\begin{enumerate}
\item\label{item:sym-orbit} A submodular function \kl{invariant} under a permutation group has a
minimizer that is a union of orbits.
\item\label{item:sym-average} Let $r\colon2^E\to\mathbb R$ be submodular and
\kl{invariant} under $\Gamma$, let $S\subseteq E$, and let $\gamma$ be selected
uniformly at random from $\Gamma$. Then the orbit average
$x\coloneqq\mathbb E_\gamma\mathbf 1_{\gamma(S)}\in[0,1]^E$ has entries
$x_e=|S\cap\Gamma(e)|/|\Gamma(e)|$, and its
\kl[Lovasz extension]{Lov\'asz extension} value satisfies $\hat r(x)\leq r(S)$.
\end{enumerate}
\end{lemma}
\begin{proof}
For \cref{item:sym-orbit}, unions and intersections of minimizers remain minimizing.
The union of all group translates of a minimizer is therefore an
\kl{invariant} minimizer.
For \cref{item:sym-average}, let $\delta\in\Gamma$. Since
$\delta^{-1}\gamma$ is again uniform in $\Gamma$,
$x_{\delta(e)}=\Pr_\gamma[e\in\delta^{-1}\gamma(S)]=x_e$, so $x$ is
constant on each orbit. Since every $\gamma$ maps each orbit to
itself, $\sum_{e'\in\Gamma(e)}x_{e'}=\mathbb E_\gamma|\gamma(S)\cap\Gamma(e)|
=|S\cap\Gamma(e)|$, and hence $x_e=|S\cap\Gamma(e)|/|\Gamma(e)|$.
Thus, $\gamma(S)$ is a random set with marginals $x$, and  the submodularity and 
invariance of $r$ gives
$\hat r(x)\leq\mathbb E_\gamma\,r(\gamma(S))=r(S)$.
\end{proof}

The following lemma is the configuration-LP counterpart
of~\itemref{lem:shared-symmetry}{item:sym-orbit}.
\begin{lemma}[Symmetric configuration-LP solutions]\label{lem:symmetric-lp}
Let $(\mathcal C,\mathcal F,\dist,g)$ be an \kl{instance} of
\kl{submodular allocation}, and let $\Gamma$ be a permutation group on
$\mathcal C$ such that $g$ is \kl{invariant} under $\Gamma$ and
$\dist(\gamma(c),f)=\dist(c,f)$ for all $\gamma\in\Gamma$, $c\in\mathcal C$,
and $f\in\mathcal F$. If the
\kl{configuration LP}~\eqref{eq:abstract-config-lp} admits a solution of finite cost (i.e., satisfying $\dist(c,f)<\infty$ whenever $c\in C$ and $y_{f,C}>0$), then it
admits an optimal solution whose support consists only of unions of
orbits.
\end{lemma}
\begin{proof}
Let $y$ be an optimal solution, which exists since the LP, restricted to variables $y_{f,C}$ where $\dist(c,f)<\infty$ for all $c\in C$, is feasible,
has finitely many variables, and has a nonnegative objective. For a configuration $C$, let
$\bar x_C\in[0,1]^{\mathcal C}$ have entries
$(\bar x_C)_c\coloneqq|C\cap\Gamma(c)|/|\Gamma(c)|$, and put
$x_f\coloneqq\sum_C y_{f,C}\,\bar x_C$ for $f\in\mathcal F$. Each $x_f$ is
constant on every orbit. For $t$ uniform in $(0,1]$, give each
\intro[abstract threshold sets]{threshold set}
$D=\{c:(x_f)_c\geq t\}$ at $f$ the mass
$y'_{f,D}\coloneqq\Pr_t[\{c:(x_f)_c\geq t\}=D]$. These sets are unions of
orbits. Each is contained in the union of the orbits that
meet a configuration $C$ at $f$ with $y_{f,C}>0$, so it has finite
connection cost to $f$ and is a configuration at $f$. The masses at
each facility sum to one.

A client $c$ is covered at $f$ with mass $\Pr_t[(x_f)_c\geq t]=(x_f)_c$,
hence in total with mass
\[
 \sum_{f\in\mathcal F}(x_f)_c
 =\frac1{|\Gamma(c)|}\sum_{c'\in\Gamma(c)}\sum_{f\in\mathcal F}\sum_{C\ni c'}y_{f,C}
 =1.
\]
By the definition of the \kl[Lovasz extension]{Lov\'asz extension}, its
convexity, and \itemref{lem:shared-symmetry}{item:sym-average}, the
facility cost at $f$ is
\[
 \sum_D y'_{f,D}\,g(D)=\hat g(x_f)
 \leq\sum_C y_{f,C}\,\hat g(\bar x_C)
 \leq\sum_C y_{f,C}\,g(C).
\]
Since $\dist(\cdot,f)$ is constant on every orbit, the connection cost at
$f$ is
\[
 \sum_D y'_{f,D}\sum_{c\in D}\dist(c,f)
 =\sum_{c\colon(x_f)_c>0}(x_f)_c\dist(c,f)
 =\sum_C y_{f,C}\sum_{c\in C}\dist(c,f).
\]
Summing over $f$, the solution $y'$ costs at most as much as $y$, so it is optimal.
\end{proof}

\section{From symmetry gaps to hardness}\label{sec:symmetry-gap}

The symmetry-gap theorem starts from an assignment instance whose optimum is much
smaller than the best assignment that sends each symmetry orbit to one
facility. We turn this gap into an oracle lower bound
(\cref{thm:symmetry-gap}) and configuration-LP gaps
(\cref{prop:abstract-config-lp}).
The oracle construction
replicates clients and hides which input client each copy represents.

A reader interested only in the application to SFL may skip the proofs
in this section, since later sections only use
\cref{def:fibered-instance},
\cref{thm:symmetry-gap}, and
\cref{prop:abstract-instances,lem:abstract-optima,prop:abstract-config-lp}.

\subsection{Symmetric instances and the symmetry-gap theorem}

\begin{definition}[Symmetric assignment instance]\label{def:fibered-instance}
A \intro[abstract symmetric instance]{symmetric assignment instance}
$(\mathcal C_0,\mathcal F,\dist,\AbstractInputCost,\AbstractGroup)$ consists of
an instance $(\mathcal C_0,\mathcal F,\dist,\AbstractInputCost)$ of
\kl{submodular allocation} (\cref{def:allocation}), with
\intro[abstract input facility cost function]{facility cost function}
$\intro*\AbstractInputCost$, and a
\intro[abstract symmetry group]{permutation group}
$\intro*\AbstractGroup$ on $\mathcal C_0$ under which both costs are
\kl{invariant}:
\begin{equation}\label{eq:fibered-invariance}
 \AbstractInputCost(\gamma(S))=\AbstractInputCost(S)
 \quad\text{and}\quad
 \dist(\gamma(u),f)=\dist(u,f)
 \qquad \forall \gamma\in\AbstractGroup,\ S\subseteq\mathcal C_0,\ u\in\mathcal C_0,\ f\in\mathcal F.
\end{equation}

We use $(\intro*\AbstractOrbit_i)_{i\in I}$ for the
\intro[abstract orbit]{orbits} of $\AbstractGroup$ and for $u \in \mathcal{C}_0$ we let
$\intro*\AbstractOrbitIndex(u) \in I$ be the index of the orbit containing $u$. For $J \subseteq I$ we set $\AbstractOrbit(J)\coloneqq\bigcup_{i\in J}\AbstractOrbit_i$. By~\eqref{eq:fibered-invariance}, $\dist(u,f)$ depends
only on $\AbstractOrbitIndex(u)$. In slight abuse of notation we write $\dist(i,f)$ for its value on
$\AbstractOrbit_i$. We assume
\begin{equation}\label{item:fibered-single}
 \AbstractInputCost(\AbstractOrbit_i)\leq|\AbstractOrbit_i|
 \qquad(i\in I).
\end{equation}

Let $\intro*\AbstractInputOpt$ be the
\intro[abstract input optimum]{unrestricted optimum} and
$\intro*\AbstractSymOpt$ the
\intro[abstract symmetric optimum]{symmetric optimum}, the minimum cost
over \intro{whole-orbit} assignments, which send each orbit to a single
facility. For $\AbstractInputOpt>0$, the ratio
$\AbstractSymOpt/\AbstractInputOpt$ is the \intro{symmetry gap}.
\end{definition}

Condition~\eqref{item:fibered-single}
holds for matroid rank functions, and we use it to obtain matroid-rank facility
cost functions in the replicated instances. The main result of this
section is the following.

\begin{theorem}[Symmetry-gap theorem]\label{thm:symmetry-gap}
Fix a \kl[abstract symmetric instance]{symmetric assignment instance}
$(\mathcal C_0,\mathcal F,\dist,\AbstractInputCost,\AbstractGroup)$, and let $\alpha\geq1$ and $\nu>0$
satisfy
\begin{equation}\label{eq:symmetry-gap-condition}
 \alpha\bigl(\AbstractInputOpt+\nu\bigr)<\AbstractSymOpt.
\end{equation}
Then for every positive integer $R$ there is a family of instances, each
on $R$ copies of every client of $\mathcal C_0$, with the following
property. For every randomized algorithm that makes at most $Q$ value
queries on every run, there is an instance in the family on which,
over its internal randomness, the algorithm returns an assignment of
cost at most $\alpha\OPT$ with probability at most
\[
 2(Q+|\mathcal F|)|\mathcal C_0|
 \exp\bigl(-2R\nu^2/|\mathcal C_0|^2\bigr).
\]
If $R\AbstractInputCost$ is integer-valued and $R\nu/|\mathcal C_0|$ is an
integer, then the facility cost functions of all instances in the
family are matroid rank functions.
\end{theorem}

\subsection{Replicated instances}
\label{sec:replication}

Fix a \kl[abstract symmetric instance]{symmetric assignment instance}
$(\mathcal C_0,\mathcal F,\dist,\AbstractInputCost,\AbstractGroup)$ with
\kl[abstract input optimum]{unrestricted optimum} $\AbstractInputOpt$ and
\kl[abstract symmetric optimum]{symmetric optimum} $\AbstractSymOpt$, a positive integer
\intro[abstract replication factor]{replication factor}
$\intro*\AbstractCopies$ and a \intro[abstract penalty weight]{penalty weight}
$\intro*\AbstractPenaltyWeight>0$.
\intro[abstract replication]{Replication} replaces each orbit
$\AbstractOrbit_i$ by a \intro[abstract public block]{public block}
$\intro*\AbstractBlock(i)$ of $\AbstractCopies|\AbstractOrbit_i|$ clients,
all with connection costs $\dist(i,\cdot)$. Put
$\mathcal C\coloneqq\bigcup_{i\in I}\AbstractBlock(i)$ and
$\AbstractBlock(J)\coloneqq\bigcup_{i\in J}\AbstractBlock(i)$.
Independently in each block, choose uniformly at random a partition of $\AbstractBlock(i)$ into
\intro[abstract copy class]{copy classes} of size $\AbstractCopies$,
one for each input client in $\AbstractOrbit_i$.
The \intro[abstract hidden partition]{hidden partition}
$\intro*\AbstractLabelMap\colon\mathcal C\to\mathcal C_0$ maps each
\intro[abstract copy]{copy} to its input \intro[abstract label]{label};
thus $\AbstractLabelMap^{-1}(u)$ is the copy class of $u$.
All data except $\AbstractLabelMap$ are public, i.e., available to the algorithm.

Define the \intro[abstract split penalty]{split penalty}
$\intro*\AbstractSplitPen\colon2^{\mathcal C_0}\to\mathbb Z_{\geq0}$ by
\[
 \AbstractSplitPen(S)\coloneqq\sum_{i\in I}
 \min\bigl\{|S\cap\AbstractOrbit_i|,\,|\AbstractOrbit_i\setminus S|\bigr\}.
\]
This is the minimum number of elements to add or remove to make
$S$ a union of orbits.

\begin{lemma}\label{prop:split_function_submodular}
The split penalty function is normalized and submodular.  
\end{lemma}
\begin{proof}
  Clearly $\psi(\varnothing)=0$, showing normalization.
  For submodularity, we express the $i$-th summand as 
  \[\min\{|S\cap P_i|,|P_i\setminus S|\} = \min\{|S\cap P_i|,|P_i|-|S\cap P_i|\} = f(|S\cap P_i|),\]
  where $f$ denotes the concave function given by $f(x)=\min\{x,|P_i|-x\}$.
  By \cref{fact:concave-cardinality}, each summand, and hence $\psi$, is submodular.
\end{proof}

For $C\subseteq\mathcal C$, define the
\intro[abstract hidden function]{hidden function} and
\intro[abstract reference function]{reference function} by
\begin{align}
 \AlignedIntro{\AbstractHidCost}(C)
 &\coloneqq\min_{S\subseteq\mathcal C_0}
 \bigl\{|C\setminus\AbstractLabelMap^{-1}(S)|
       +\AbstractCopies\,\AbstractInputCost(S)
       +\AbstractPenaltyWeight\AbstractSplitPen(S)\bigr\},
 \label{eq:abstract-hidden-function}\\
 \AlignedIntro{\AbstractRefCost}(C)
 &\coloneqq\min_{J\subseteq I}
 \bigl\{|C\setminus\AbstractBlock(J)|
       +\AbstractCopies\,\AbstractInputCost(\AbstractOrbit(J))\bigr\}.
 \label{eq:abstract-reference-function}
\end{align}
The reference minimum restricts the hidden minimum to witnesses
$S=\AbstractOrbit(J)$, for which the penalty vanishes and
$\AbstractLabelMap^{-1}(S)=\AbstractBlock(J)$. It therefore depends only
on public data. The hidden function also allows witnesses that split
orbits. The \kl[abstract split penalty]{split penalty} and the \kl[abstract penalty weight]{penalty weight} $\AbstractPenaltyWeight$ are chosen so that
such splits pay off only for queries sufficiently correlated with the
hidden labels.

The \intro[abstract hidden instance]{hidden instance}
$(\mathcal C,\mathcal F,\dist,\AbstractHidCost)$ and
\intro[abstract reference instance]{reference instance}
$(\mathcal C,\mathcal F,\dist,\AbstractRefCost)$ use these
facility cost functions on the replicated clients, with the same facilities and
same connection costs. Write $\intro*\AbstractHidOpt$ and
$\intro*\AbstractRefOpt$ for their optima, which we refer to as
\intro[abstract hidden optimum]{hidden optimum} and
\intro[abstract reference optimum]{reference optimum}.

We collect and prove several properties.

\begin{lemma}[Replicated instances]\label{prop:abstract-instances}
\leavevmode
\begin{enumerate}
\item\label{item:abstract-functions}
The functions $\AbstractHidCost$ and $\AbstractRefCost$ are normalized,
monotone, and submodular.
\item\label{item:abstract-matroid}
If $\AbstractCopies\AbstractInputCost$ is integer-valued, then
$\AbstractRefCost$ is a matroid rank function. If in addition $\AbstractPenaltyWeight$ is an integer, then $\AbstractHidCost$ also a matroid rank function.
\item\label{item:abstract-comparison}
The function $\AbstractRefCost$ is independent of $\AbstractLabelMap$, and
$\AbstractHidCost\leq\AbstractRefCost$.
\end{enumerate}
\end{lemma}

\begin{proof}
For \cref{item:abstract-functions}, apply \itemref{lem:shared-minimization}{item:min-rank} to
$\AbstractHidCost$ with the map $\varphi=\AbstractLabelMap$ and function
$r=\AbstractCopies\AbstractInputCost+\AbstractPenaltyWeight\AbstractSplitPen$.
Note that by \cref{prop:split_function_submodular}, $r$ is normalized and submodular as a nonnegative linear combination of normalized submodular functions.
For $\AbstractRefCost$, use the map $\varphi=\AbstractOrbitIndex\circ\AbstractLabelMap$, which sends every copy in $\AbstractBlock(i)$ to $i$ 
and the function $r(J)=\AbstractCopies\,\AbstractInputCost(\AbstractOrbit(J))$. We have $r(\varnothing)=\AbstractInputCost(\varnothing)=0$. Moreover, $r$ is submodular as $\AbstractInputCost$ is and $J\mapsto\AbstractOrbit(J)$ preserves
unions and intersections.

For \cref{item:abstract-matroid}, \itemref{lem:shared-minimization}{item:min-rank} further tells us that if $\AbstractCopies\AbstractInputCost$ is integral, then $\AbstractRefCost$ is a matroid rank function, and that if both $\AbstractCopies\AbstractInputCost$ and $\AbstractPenaltyWeight$ are integral, then so is $\AbstractHidCost$.

For \cref{item:abstract-comparison}, the definition of the reference function does not depend on the hidden partition. Compared to the hidden function, it
restricts the minimum to unions of orbits, because $\AbstractSplitPen(P(J))=0$ for every $J\subseteq I$. So
$\AbstractHidCost(C)\leq\AbstractRefCost(C)$ for all $C\subseteq\mathcal C$.
\end{proof}

\begin{lemma}[Replicated optima]\label{lem:abstract-optima}
\leavevmode
\begin{enumerate}
\item\label{item:abstract-ref-opt}
We have $\AbstractRefOpt=\AbstractCopies\,\AbstractSymOpt$.
\item\label{item:abstract-hid-opt}
We have
$\AbstractHidOpt\leq\AbstractCopies\,\AbstractInputOpt+\AbstractPenaltyWeight|\mathcal C_0|$.
\end{enumerate}
\end{lemma}

\begin{proof}
For \cref{item:abstract-ref-opt}, consider an optimum reference assignment $(C_f)_{f\in\mathcal F}$, which, among all optimum reference assignments, minimizes $\sum_{i\in I}|\{f\colon C_f\cap\AbstractBlock(i)\neq\varnothing\}|$, i.e., which `splits public blocks as little as possible'. Our goal is to show that each $C_f$ is a union of blocks. Assume towards a contradiction that some $\AbstractBlock(i)$ sends some of its clients to facility $f_1$ and some of its clients to facility $f_2$. Let $(J_f)_{f\in\mathcal F}$ be minimizing witnesses for $(C_f)_{f\in \mathcal F}$ in~\eqref{eq:abstract-reference-function}. Consider changing the assignment by moving all clients in $\AbstractBlock(i)\cap C_{f_1}$ to $f_2$. This yields an assignment of cost at most
\[
 \AbstractRefOpt+|\AbstractBlock(i)\cap C_{f_1}|\cdot\bigl(\dist(i,f_2)-\dist(i,f_1)+\mathbf 1_{\{i\in J_{f_1}\}}-\mathbf 1_{\{i\in J_{f_2}\}}\bigr).
\]
Conversely, moving all clients in $\AbstractBlock(i)\cap C_{f_2}$ to $f_1$ gives an assignment of cost at most
\[
 \AbstractRefOpt+|\AbstractBlock(i)\cap C_{f_2}|\cdot\bigl(\dist(i,f_1)-\dist(i,f_2)+\mathbf 1_{\{i\in J_{f_2}\}}-\mathbf 1_{\{i\in J_{f_1}\}}\bigr).
\]
At least one of the new assignments is at most as expensive as the original one, however, both of them yield a smaller value of $\sum_{i\in I}|\{f\colon C_f\cap\AbstractBlock(i)\neq\varnothing\}|$, a contradiction. Hence, each $C_f$ is a union of blocks $\AbstractBlock(i)$, i.e., we may restrict ourselves to assignments of this form.

Fix $f \in \mathcal F$ and let $H \subseteq I$ be such that $C_f=\AbstractBlock(H)$. Then, for any $J\subseteq I$,
 \begin{align}
 |\AbstractBlock(H)\setminus\AbstractBlock(J)|+\AbstractCopies\,\AbstractInputCost(\AbstractOrbit(J))
 &\geq|\AbstractBlock(H\setminus J)|+\AbstractCopies\,\AbstractInputCost(\AbstractOrbit(J\cap H)) \tag{by monotonicity of $g$}\\
 &=\AbstractCopies\,|\AbstractOrbit(H\setminus J)|+\AbstractCopies\,\AbstractInputCost(\AbstractOrbit(J\cap H)) \notag\\
 &\geq\AbstractCopies\,\AbstractInputCost(\AbstractOrbit(H\setminus J))+\AbstractCopies\,\AbstractInputCost(\AbstractOrbit(J\cap H)) \tag{by \eqref{item:fibered-single} and subadditivity}\\
 &\geq\AbstractCopies\,\AbstractInputCost(\AbstractOrbit(H)) \tag{by subadditivity of $g$},
 \end{align}
so $\AbstractRefCost(\AbstractBlock(H)) \ge \AbstractCopies\,\AbstractInputCost(\AbstractOrbit(H))$. 

Hence, the facility cost of each whole-block assignment for the reference instance is at least $\AbstractCopies$ times that of the corresponding whole-orbit assignment for the symmetric instance. Connection costs scale by exactly $\AbstractCopies$, since each block $\AbstractBlock(i)$ consists of $\AbstractCopies|\AbstractOrbit_i|$ clients with connection costs $\dist(i,\cdot)$. As some optimal reference assignment is a whole-block assignment, $\AbstractRefOpt \geq \AbstractCopies \, \AbstractSymOpt$.

Conversely, replace each cluster $\AbstractOrbit(H_f)$ of a whole-orbit
assignment $(H_f)_{f \in \mathcal F}$ by $\AbstractBlock(H_f)$. The witness $H_f$ bounds its
facility cost by $\AbstractCopies\AbstractInputCost(\AbstractOrbit(H_f))$,
and its connection cost scales by $\AbstractCopies$, hence $\AbstractRefOpt \leq \AbstractCopies \, \AbstractSymOpt$. This proves equality.

For \cref{item:abstract-hid-opt}, let $\sigma$ be an optimal assignment for the instance $(\mathcal C_0,\mathcal F,\dist,\AbstractInputCost)$ achieving objective value $\AbstractInputOpt$.
Let $S_f=\sigma^{-1}(f)$, and send
$\AbstractLabelMap^{-1}(S_f)$ to $f$. The witness $S_f$ in
\eqref{eq:abstract-hidden-function} leaves no copy uncovered, so
\[
 \AbstractHidCost(\AbstractLabelMap^{-1}(S_f))
 \leq\AbstractCopies\AbstractInputCost(S_f)
       +\AbstractPenaltyWeight\AbstractSplitPen(S_f).
\]
Connection costs again scale by $\AbstractCopies$. Summing over $f$
and using $\AbstractSplitPen(S_f)\leq|S_f|$ and
$\sum_f|S_f|=|\mathcal C_0|$ proves the claimed bound.
\end{proof}

\subsection{The oracle lower bound}
\label{sec:oracle-lower-bound}

We now prove \cref{thm:symmetry-gap}. We let the replication factor
$\AbstractCopies$ equal the corresponding parameter of \cref{thm:symmetry-gap} and let
$\AbstractPenaltyWeight\coloneqq \AbstractCopies\nu/|\mathcal C_0|$ (where $\nu$ is the parameter from \cref{thm:symmetry-gap}). The family consists of the
hidden instances $(\mathcal C,\mathcal F,\dist,\AbstractHidCost)$, which differ only in $\AbstractLabelMap$.
The proof compares the algorithm's queries against a run with the
public reference function. Every fixed query receives the same answer
from both functions with high probability over the random hidden
partition. Once we condition on the algorithm's random bits, its
queries against the reference function are fixed.
The gap between their optima from \itemref{lem:abstract-optima}{item:abstract-ref-opt} and~\itemref*{lem:abstract-optima}{item:abstract-hid-opt} then implies that the algorithm fails.

For a query $C\subseteq\mathcal C$ and a client $u\in\mathcal C_0$, define its
\intro[abstract copy count]{copy counts},
\intro[abstract block average]{block averages}, and
\intro[abstract imbalance]{imbalance} by
\[
 \intro*\AbstractCopyCount_u(C)\coloneqq|C\cap\AbstractLabelMap^{-1}(u)|,
 \qquad
 \intro*\AbstractBlockAvg_u(C)\coloneqq
 \frac{|C\cap\AbstractBlock(\AbstractOrbitIndex(u))|}
 {|\AbstractOrbit_{\AbstractOrbitIndex(u)}|},
 \qquad
 \intro*\AbstractImbalance_u(C)\coloneqq\AbstractCopyCount_u(C)-\AbstractBlockAvg_u(C).
\]
The block averages are constant on each orbit, so
$\AbstractBlockAvg_u(C) = \AbstractBlockAvg_{u'}(C)$ for all $u'$ in the same orbit as $u$. 
Note furthermore that the imbalance sums to zero on each orbit, i.e., $\sum_{u\in\AbstractOrbit_i}\AbstractImbalance_u(C) = 0$ for all $i \in I$. Call $C$
\intro[abstract well balanced]{well-balanced} if
$\|\AbstractImbalance(C)\|_\infty \coloneqq \max_{u\in\mathcal C_0}|\AbstractImbalance_u(C)| \leq\AbstractPenaltyWeight$.
\cref{fig:abstract-replication} illustrates these quantities.

\begin{figure}[!htbp]
\centering
\begin{tikzpicture}[
  x=1cm,y=1cm,
  every node/.style={font=\small},
  panel/.style={font=\small\bfseries,align=center},
  subpanel/.style={font=\scriptsize,text=black!70,align=center},
  boundary/.style={draw=black!60,densely dashed,rounded corners=4pt},
  seed/.style={circle,draw=black!70,minimum size=5.5mm,inner sep=0pt},
  copy/.style={circle,draw=black!70,line width=.35pt,
    minimum size=4.2mm,inner sep=0pt,font=\scriptsize},
  query/.style={draw=definitioncolor,line width=1.1pt},
  count/.style={font=\scriptsize,align=center},
  op/.style={-stealth,thick,draw=black!65},
]
\node[panel] at (1.65,0) {(a) Input \kl[abstract orbit]{orbit}};
\node[panel] at (6.70,0) {(b) \kl[abstract label]{Hidden labels} revealed};
\node[panel] at (11.85,0) {(c) Public view};
\node[subpanel] at (1.65,-.38) {before \kl[abstract replication]{replication}};
\node[subpanel] at (6.70,-.38) {shown for explanation};
\node[subpanel] at (11.85,-.38) {the same replicated clients};

\node at (1.65,-.90) {$\AbstractOrbit_1$};
\draw[boundary] (0,-1.16) rectangle (3.3,-2.76);
\foreach \j/\x/\shade/\ink in {
  1/.65/cyc0/black,2/1.65/cyc1/black,3/2.65/cyc2/white}
  \node[seed,fill=\shade,text=\ink] at (\x,-1.95) {$u_\j$};
\node[count] at (1.65,-3.08) {$|\AbstractOrbit_1|=3$};

\draw[op] (3.52,-1.92)--(4.78,-1.92)
  node[midway,above,font=\scriptsize] {$\AbstractCopies=4$ \kl[abstract copy]{copies}};
\draw[op] (4.78,-2.44)--(3.52,-2.44)
  node[midway,below,font=\scriptsize] {$\AbstractLabelMap\colon\mathcal C\to\mathcal C_0$};
\draw[op] (8.60,-1.92)--(9.94,-1.92)
  node[midway,above,font=\scriptsize] {hide $\AbstractLabelMap$}
  node[midway,below,font=\scriptsize] {same clients};

\foreach \shift in {0,5.15}{
  \node at ({6.70+\shift},-.90) {$\AbstractBlock(1)$};
  \draw[boundary] ({5.0+\shift},-1.16)
    rectangle ({8.4+\shift},-3.02);
}

\foreach \i/\j/\selected in {
  1/1/1,2/2/1,3/3/1,4/2/0,
  5/3/0,6/1/0,7/2/1,8/3/1,
  9/2/0,10/3/1,11/1/0,12/1/0}{
  \pgfmathsetmacro\px{5.48+mod(\i-1,4)*.81}
  \pgfmathsetmacro\py{-1.52-floor((\i-1)/4)*.56}
  \ifnum\j=1 \def\shade{cyc0}\def\ink{black}\fi
  \ifnum\j=2 \def\shade{cyc1}\def\ink{black}\fi
  \ifnum\j=3 \def\shade{cyc2}\def\ink{white}\fi
  \node[copy,fill=\shade,text=\ink] at (\px,\py) {$c_{\i}$};
  \node[copy,fill=white] at ({\px+5.15},\py) {$c_{\i}$};
  \foreach \shift in {0,5.15}{
    \ifnum\selected=1
      \draw[query] ({\px+\shift},\py) circle[radius=.265];
    \fi
  }
}
\node[count] at (6.70,-3.36)
  {$\AbstractCopyCount(C)|_{\AbstractOrbit_1}=(1,2,3)$};
\node[count] at (11.85,-3.28) {$|C\cap\AbstractBlock(1)|=6$};
\node[count] at (11.85,-3.62)
  {$\AbstractBlockAvg(C)|_{\AbstractOrbit_1}=(2,2,2)$};

\draw[query] (3.80,-4.10) circle[radius=.10];
\node[anchor=west,font=\scriptsize] at (4.02,-4.10)
  {ring: $c\in C$\qquad shading: hidden \kl[abstract label]{label} $\AbstractLabelMap(c)$};
\end{tikzpicture}
\caption{\kl[abstract replication]{Replication} in the hiding construction, shown for one \kl[abstract orbit]{orbit} of
size three and $\AbstractCopies=4$, so each input client has four \kl[abstract copy]{copies}.
Panels~(b) and~(c) show the same clients in the same positions.
In~(b) the shading gives the input \kl[abstract label]{labels} and so reveals $\AbstractLabelMap$; in~(c)
the labels have been removed. Both views show the red rings that mark
the query $C\subseteq\AbstractBlock(1)$.
The hidden \kl[abstract copy count]{copy counts} of $C$ are $(1,2,3)$, and the \kl[abstract public block]{public block} total
gives the averages $(2,2,2)$. Thus
$\|\AbstractCopyCount(C)-\AbstractBlockAvg(C)\|_\infty=1$, and the query is \kl[abstract well balanced]{well-balanced} for
$\AbstractPenaltyWeight\geq1$.}
\label{fig:abstract-replication}
\end{figure}

\Needspace{9\baselineskip}
The key observation is that the functions $\AbstractHidCost$ and $\AbstractRefCost$ agree on all
well-balanced queries. Recall that the parameter $\AbstractPenaltyWeight$ affects both the definition of well-balancedness and the definition of $\AbstractHidCost$.
\begin{lemma}[Agreement on balanced queries]\label{lem:abstract-agreement}
For every hidden partition $\AbstractLabelMap$, the functions $\AbstractHidCost$ and $\AbstractRefCost$ agree on every
well-balanced query.
\end{lemma}
\begin{proof}
Fix such a query $C\subseteq\mathcal C$. We first express the reference
function in terms of the block averages and then drop the restriction to
unions of orbits:
\begin{equation}\label{eq:abstract-symmetrized}
 \begin{aligned}
 \AbstractRefCost(C)
 &=\min_{J\subseteq I}
 \left(|C\setminus\AbstractBlock(J)|
       +\AbstractCopies\AbstractInputCost(\AbstractOrbit(J))\right)\\
 &=\min_{J\subseteq I}
 \left(\sum_{i\in I}|\AbstractOrbit_i\setminus\AbstractOrbit(J)|\,
       \frac{|C\cap\AbstractBlock(i)|}{|\AbstractOrbit_i|}
       +\AbstractCopies\AbstractInputCost(\AbstractOrbit(J))\right)\\
 &=\min_{S\subseteq\mathcal C_0}
 \left(\sum_{i\in I}|\AbstractOrbit_i\setminus S|\,
       \frac{|C\cap\AbstractBlock(i)|}{|\AbstractOrbit_i|}
       +\AbstractCopies\AbstractInputCost(S)\right)\\
 &=\min_{S\subseteq\mathcal C_0}
 \left(\sum_{u\notin S}\AbstractBlockAvg_u(C)
       +\AbstractCopies\AbstractInputCost(S)\right).
 \end{aligned}
\end{equation}
The first equality is the definition~\eqref{eq:abstract-reference-function}.
The second holds because $|\AbstractOrbit_i\setminus\AbstractOrbit(J)|$
is $|\AbstractOrbit_i|$ for $i\notin J$ and $0$ for $i\in J$.
For the third, the objective on the right depends on $S$ only through
$\AbstractInputCost(S)$ and the number of clients of each orbit that $S$
misses. It is therefore invariant under $\AbstractGroup$, and it is
submodular, so by \itemref{lem:shared-symmetry}{item:sym-orbit} some
minimizer is a union of orbits $\AbstractOrbit(J)$. Restricting the minimum
to such sets thus does not change it. The last equality only rewrites the
sum, since $\AbstractBlockAvg_u(C)=|C\cap\AbstractBlock(i)|/|\AbstractOrbit_i|$
for $u\in\AbstractOrbit_i$.

For every witness $S$, the bound $|\AbstractImbalance_u(C)|\leq\AbstractPenaltyWeight$
and the zero sum on each orbit imply
\[
 \sum_{u\in S\cap\AbstractOrbit_i}\AbstractImbalance_u(C)
 =-\sum_{u\in\AbstractOrbit_i\setminus S}\AbstractImbalance_u(C)
 \leq\AbstractPenaltyWeight
 \min\bigl\{|S\cap\AbstractOrbit_i|,\,|\AbstractOrbit_i\setminus S|\bigr\}.
\]
Summing over orbits yields
\begin{equation}\label{eq:abstract-imbalance-absorbed}
 \sum_{u\in S}\AbstractImbalance_u(C)
 =\sum_{i\in I}\sum_{u\in S\cap\AbstractOrbit_i}\AbstractImbalance_u(C)
 \leq\AbstractPenaltyWeight\AbstractSplitPen(S).
\end{equation}
Consequently, every witness in the hidden minimum has cost at least
the reference minimum:
\begin{equation}\label{eq:abstract-count-decomposition}
 \begin{aligned}
 |C\setminus\AbstractLabelMap^{-1}(S)|
 +\AbstractCopies\AbstractInputCost(S)
 +\AbstractPenaltyWeight\AbstractSplitPen(S)
  &=\sum_{u\notin S}\AbstractCopyCount_u(C)
   +\AbstractCopies\AbstractInputCost(S)
   +\AbstractPenaltyWeight\AbstractSplitPen(S)
  \\
 &=\sum_{u\notin S}\AbstractBlockAvg_u(C)
   +\AbstractCopies\AbstractInputCost(S)
   +\AbstractPenaltyWeight\AbstractSplitPen(S)
   -\sum_{u\in S}\AbstractImbalance_u(C)\\
 &\geq\sum_{u\notin S}\AbstractBlockAvg_u(C)
   +\AbstractCopies\AbstractInputCost(S)
 \geq\AbstractRefCost(C).
 \end{aligned}
\end{equation}
Minimizing over $S$ yields $\AbstractHidCost(C)\geq\AbstractRefCost(C)$. With \itemref{prop:abstract-instances}{item:abstract-comparison} for
the reverse inequality we obtain the desired equality.
\end{proof}
For the next lemma, we need the following result by Hoeffding~\cite[Theorem~2 and Section~6]{Hoeffding1963}:
\begin{theorem}[\cite{Hoeffding1963}]\label{theorem:Hoeffding}
Let $\alpha_1,\dots,\alpha_N$ be real numbers and let $X_1,\dots,X_n$ be drawn from the multiset $\{\alpha_1,\dots,\alpha_N\}$ uniformly at random without replacement. Let $S\coloneqq \frac{1}{n}\cdot \sum_{i=1}^n X_i$ denote their average, let $\mu\coloneqq \frac{1}{N}\sum_{i=1}^N \alpha_i$, $a\coloneqq \min_{i=1}^N \alpha_i$ and $b\coloneqq \max_{i=1}^N \alpha_i$. Then, for $t\ge 0$,
\[\Pr[S-\mu \ge t] \le \exp(-2t^2n/(b-a)^2).\]
\end{theorem}
\begin{corollary}\label{cor:Hoeffding}
In the setting of \cref{theorem:Hoeffding}, we have
$\Pr[|S-\mu| \ge t] \le 2\cdot\exp(-2t^2n/(b-a)^2)$.
\end{corollary}
\begin{proof}
 This follows by negating all $\alpha_i$ to obtain the same bound on the probability that $\mu-S\ge t$, and then applying a union bound. 
\end{proof}
\begin{lemma}[Agreement on a fixed query]\label{lem:abstract-fixed-query}
Let $\AbstractLabelMap$ be the random hidden partition of
\cref{sec:replication}. For every fixed query
$C\subseteq\mathcal C$,
\[
 \Pr_{\AbstractLabelMap}[\AbstractHidCost(C)\neq\AbstractRefCost(C)]
 \leq\Pr_{\AbstractLabelMap}[C\text{ is not well-balanced}]
 \leq2\,|\mathcal C_0|\exp(-2\AbstractPenaltyWeight^2/\AbstractCopies).
\]
\end{lemma}
\begin{proof}
The first inequality is \cref{lem:abstract-agreement}.
Each copy class $\AbstractLabelMap^{-1}(u)$ is a uniformly distributed
subset of cardinality $\AbstractCopies$ in its public block $B(\omega(u))$.
This is equivalent to sampling $R$ elements $c_1,\dots,c_R$ uniformly without replacement from $B(\omega(u))$.
We identify each element of $c\in B(\omega(u))\cap C$ with the number $\alpha_c=1$ and every $c\in B(\omega(u))\setminus C$ with $\alpha_c=0$. 
Let $X_i$ denote the random variable $\alpha_{c_i}$ and let $S = R^{-1}\sum_{i=1}^R X_i$ denote the average of the $X_i$. The minimum and maximum value of the $\alpha_c$'s are $a=0$ and $b=1$.
Then $\mu_u(C)=\sum_{i=1}^R X_i=R\cdot S$ is the number of elements from $C$ that were sampled for $\pi^{-1}(u)$.
Moreover, we have 
\begin{equation}
  \mu\coloneqq \frac{1}{|B(\omega(u))|}\cdot \sum_{c\in B(\omega(u))} \alpha_c = \frac{|C\cap B(\omega(u))|}{|B(\omega(u))|}=\frac{1}{R}\cdot \AbstractBlockAvg_u(C)\,.
\end{equation}

Thus, \cref{cor:Hoeffding} yields
\[
\begin{aligned}
\Pr_{\AbstractLabelMap}[|\AbstractImbalance_u(C)|>\AbstractPenaltyWeight]
&=\Pr_{\AbstractLabelMap}[|\AbstractCopyCount_u(C)-\AbstractBlockAvg_u(C)|>\AbstractPenaltyWeight]
=\Pr_{\AbstractLabelMap}[|R\cdot S-R\cdot\mu|>\AbstractPenaltyWeight]
=\Pr_{\AbstractLabelMap}[|S-\mu|>\AbstractPenaltyWeight/R]\\
&\le 2\cdot \exp(-2(\AbstractPenaltyWeight/R)^2\cdot R)=2\cdot \exp(-2\AbstractPenaltyWeight^2/R).
\end{aligned}
\]
Union bound over $u\in\mathcal C_0$ gives the result.
\end{proof}

We are ready to prove \cref{thm:symmetry-gap}.
\begin{proof}[Proof of \cref{thm:symmetry-gap}]
The copies in $\AbstractBlock(i)$ have the connection costs of the
clients in $\AbstractOrbit_i$, and by
\itemref{prop:abstract-instances}{item:abstract-functions}
and~\itemref*{prop:abstract-instances}{item:abstract-matroid}
the facility cost functions are normalized, monotone, and submodular,
and matroid rank functions under the stated integrality conditions.

Recall that we set $\AbstractCopies=R$ and
$\AbstractPenaltyWeight=R\nu/|\mathcal C_0|$ at the beginning of
\cref{sec:oracle-lower-bound}. For the success probability, let
$\eta\coloneqq2|\mathcal C_0|\exp(-2\AbstractPenaltyWeight^2/\AbstractCopies)
=2|\mathcal C_0|\exp(-2R\nu^2/|\mathcal C_0|^2)$.
To compare the cost of the returned assignment under the
two oracles, we treat the output of the $|\mathcal F|$ clusters as additional queries to the oracle. This gives at most $Q+|\mathcal F|$ queries.

Run the algorithm against the public reference oracle and condition
on its random bits. This fixes its queries and its output, and they do
not depend on $\AbstractLabelMap$, so \cref{lem:abstract-fixed-query}
and a union bound give agreement on all of them with probability at
least $1-(Q+|\mathcal F|)\eta$ over $\AbstractLabelMap$. On this event, induction over the
queries shows that the hidden-oracle run has the same transcript and
output. The output assignment has the same cost under both oracles,
since they agree on its clusters and connection costs are public.
By \itemref{lem:abstract-optima}{item:abstract-ref-opt}
and~\itemref*{lem:abstract-optima}{item:abstract-hid-opt}
and~\eqref{eq:symmetry-gap-condition}, this cost is at least
\[
 \AbstractCopies\,\AbstractSymOpt
 >\alpha R\bigl(\AbstractInputOpt+\nu\bigr)
 =\alpha\bigl(\AbstractCopies\,\AbstractInputOpt
              +\AbstractPenaltyWeight|\mathcal C_0|\bigr)
 \geq\alpha\,\AbstractHidOpt.
\]
Thus, for every fixed choice of random bits, success has probability
at most $(Q+|\mathcal F|)\eta$ over $\AbstractLabelMap$, and hence also
over the partition and internal randomness jointly. Averaging over the partition gives
a fixed $\AbstractLabelMap$ with the same upper bound on success over
the algorithm's randomness.
\end{proof}

\subsection{Configuration-LP gaps}
\label{sec:config-lp-gaps}

We now turn the symmetry gap of the
\kl[abstract symmetric instance]{symmetric assignment instance}
$(\mathcal C_0,\mathcal F,\dist,\AbstractInputCost,\AbstractGroup)$
fixed in \cref{sec:replication} into integrality gaps of the
configuration LP~\eqref{eq:abstract-config-lp}. We give two
constructions, each derived from this instance. In both, the integral
optimum is a fixed multiple of the symmetric optimum
($\AbstractCopies\,\AbstractSymOpt$ and $\AbstractSymOpt$, respectively),
while the configuration LP is at most the same multiple of the
unrestricted optimum ($\AbstractCopies\,\AbstractInputOpt$ and
$\AbstractInputOpt$, respectively).
The first construction is the
\kl[abstract reference instance]{reference instance}
$(\mathcal C,\mathcal F,\dist,\AbstractRefCost)$ of
\cref{sec:replication}. Its optimum is
$\AbstractRefOpt=\AbstractCopies\,\AbstractSymOpt$ by
\itemref{lem:abstract-optima}{item:abstract-ref-opt}. Write
$\operatorname{ConfLP}_{\mathrm{ref}}$ for the value of its
configuration LP. The hidden partition plays no role here.

The second construction contracts each orbit to a single client. The
\intro[abstract quotient instance]{quotient instance} has client set
$I$, facilities $\mathcal F$,
\intro[abstract quotient function]{facility cost function}
$\intro*\AbstractQuotCost(J)\coloneqq\AbstractInputCost(\AbstractOrbit(J))$
for $J\subseteq I$, and connection cost $|\AbstractOrbit_i|\dist(i,f)$
from $i$ to $f$. Thus the client $i$ stands for the whole orbit
$\AbstractOrbit_i$. An assignment $\tau\colon I\to\mathcal F$ of the
quotient instance corresponds to the \kl{whole-orbit} assignment of
$(\mathcal C_0,\mathcal F,\dist,\AbstractInputCost,\AbstractGroup)$
that sends every client of $\AbstractOrbit_i$ to $\tau(i)$, and both
have the same cost. Hence the optimum of the quotient instance is
$\AbstractSymOpt$. Write $\operatorname{ConfLP}_{\mathrm{quot}}$ for
the value of its configuration LP.

\Needspace{9\baselineskip}
\begin{proposition}[Configuration-LP gaps]\label{prop:abstract-config-lp}
Every \kl[assignment]{assignment} $\sigma$ of the
\kl[abstract symmetric instance]{symmetric assignment instance}
$(\mathcal C_0,\mathcal F,\dist,\AbstractInputCost,\AbstractGroup)$
gives configuration-LP solutions of cost at most
\begin{enumerate}
\item\label{item:abstract-quotient}
$\Cost(\sigma)$ in the quotient instance;
\item\label{item:abstract-lp}
$\AbstractCopies\Cost(\sigma)$ in the reference instance.
\end{enumerate}
Consequently,
$\operatorname{ConfLP}_{\mathrm{quot}}\leq\AbstractInputOpt$
and $\operatorname{ConfLP}_{\mathrm{ref}}\leq\AbstractCopies\,\AbstractInputOpt$.
\end{proposition}

By \itemref{lem:abstract-optima}{item:abstract-ref-opt} and
the correspondence above, both instances therefore have integrality gap
at least $\AbstractSymOpt/\AbstractInputOpt$ when $\AbstractInputOpt>0$.
The reference instance keeps the connection costs $\dist$ and has
matroid-rank facility costs when $\AbstractCopies\AbstractInputCost$ is
integer-valued; \cref{sec:sfl-config-gap} uses it. The quotient
instance keeps the structure of $\AbstractInputCost$ on unions of
orbits, such as weighted coverage.
In this paper, the quotient instance serves only as an intermediate
step towards the reference instance.

\begin{proof}[Proof of \cref{prop:abstract-config-lp}]
Let $\sigma\colon\mathcal C_0\to\mathcal F$ be an \kl[assignment]{assignment} of
$(\mathcal C_0,\mathcal F,\dist,\AbstractInputCost,\AbstractGroup)$. Giving
$\sigma^{-1}(f)$ mass one at each facility $f$ is a feasible solution of
its configuration LP of cost $\Cost(\sigma)$. Since $\AbstractInputCost$
and the connection costs are invariant under $\AbstractGroup$,
\cref{lem:symmetric-lp} gives an optimal solution $y$, of cost at most
$\Cost(\sigma)$, whose support consists of unions of orbits
$\AbstractOrbit(J)$, $J\subseteq I$.

For \cref{item:abstract-quotient}, the function $\AbstractQuotCost$ is
normalized, monotone, and submodular, since $J\mapsto\AbstractOrbit(J)$
preserves inclusions, unions, and intersections. In the quotient
instance, give each $J\subseteq I$ the mass $y_{f,\AbstractOrbit(J)}$ at
$f$. Every $i\in I$ is covered with total mass one, since every
client of $\AbstractOrbit_i$ is covered with total mass one in $y$. The
facility cost $\AbstractQuotCost(J)=\AbstractInputCost(\AbstractOrbit(J))$
and the connection cost $\sum_{i\in J}|\AbstractOrbit_i|\dist(i,f)$ of $J$
at $f$ equal those of $\AbstractOrbit(J)$ at $f$. This solution therefore has cost at most
$\Cost(\sigma)$, and taking $\sigma$ optimal gives
$\operatorname{ConfLP}_{\mathrm{quot}}\leq\AbstractInputOpt$.

For \cref{item:abstract-lp}, we lift: every feasible solution $y$
of the quotient configuration LP gives a solution of the reference
configuration LP with $\AbstractCopies$ times its cost. Give the
configuration $\AbstractBlock(J)$ at $f$ the mass $y_{f,J}$. Every copy in
$\AbstractBlock(i)$ is then covered with total mass
$\sum_{f\in\mathcal F}\sum_{J\ni i}y_{f,J}=1$, and $\AbstractBlock(J)$ is a
configuration at $f$ because its copies have connection costs $\dist(i,f)$
for $i\in J$. The \kl[witness]{witness} $J$
in~\eqref{eq:abstract-reference-function} gives
$\AbstractRefCost(\AbstractBlock(J))\leq\AbstractCopies\AbstractInputCost(\AbstractOrbit(J))
=\AbstractCopies\AbstractQuotCost(J)$, and the connection cost of
$\AbstractBlock(J)$ is
$\sum_{i\in J}\AbstractCopies|\AbstractOrbit_i|\dist(i,f)$, which is
$\AbstractCopies$ times that of $J$ in the quotient instance. Applying this
to the solution of \cref{item:abstract-quotient} proves the claimed
bound and gives
$\operatorname{ConfLP}_{\mathrm{ref}}\leq\AbstractCopies\operatorname{ConfLP}_{\mathrm{quot}}
\leq\AbstractCopies\,\AbstractInputOpt$.
\end{proof}

\section{Lower bound for submodular facility location}\label{sec:sfl-lower-bound}

Since the symmetry-gap theorem takes care of the query model, the
construction in \cref{sec:quantitative-seed} is purely combinatorial. We
build a symmetric instance and compare the cost of splitting its
orbits with the cost of keeping them together.

\subsection{A symmetry gap on a cyclic lift}
\label{app:quantitative}\label{sec:quantitative-seed}

In this subsection, we present the construction outlined in
\cref{sec:overview}. Let $d$ be sufficiently large, and set $K=d^{d+3}$ and
$m=K^2$ (see~\eqref{eq:quant-parameters}).

\cref{sec:quant-source} shows that there exists a $d$-regular bipartite
graph with $m$ vertices on each side, formed by $d$ perfect matchings, that 
has cover number at least $m\log d/(2d)$ and expands on small sets in the following sense:
every set $W$ of at most $2K$ right vertices has at least $(d-2) |W|$ many neighbors. \cref{sec:quant-seed} defines its cyclic lift formally.

\cref{sec:quant-cap} defines the facility cost of a set $S$ of
clients as the minimum of $|S\cup N(W)|-(d-3)|W|$ over sets $W$ of at
most $2K$ lifts. We cap the facility cost at $K$. If the lifts in $W$ have disjoint
neighborhoods, this expression equals $|S\setminus N(W)|+3|W|$, i.e., every
client not covered by $W$ costs $1$, and every lift covers its $d$
neighbors at cost $3$. In particular, the neighborhood of a single lift
costs at most $3$. Overlapping neighborhoods make the expression
smaller, but expansion keeps small sets $W$ close to disjoint, and the
cap keeps the minimizing sets $W$ small. We will show that, together, they make the
function a matroid rank function that stays close to this cover cost. On a union of whole orbits, it is at least
$d$ times the cover number of the corresponding left vertices, up to
the cap.

\cref{sec:metric-definition} adds the metric, the shortest-path
distance scaled by $1/d$ and capped at $1$, and proves the symmetry gap
(\cref{lem:symmetry-gap}). The assignment that splits every orbit costs
at most $4m$. In a whole-orbit assignment, an orbit more than $d-1$
hops from its facility pays connection cost $d$, which accounts for one
right vertex of a cover. By $d$-regularity, each facility has at most $d^d \ll K$ orbits within
$d-1$ hops, too few for the cap to apply, so every whole-orbit
assignment costs at least $d$ times the cover number, i.e., at least
$d \cdot \frac{m\log d}{2d} = \frac{m\log d}{2}$. The gap is therefore at least $\frac{m\log d}{2}\big/(4m) = \frac{\log d}{8}$. Finally,
\cref{sec:quant-hiding} applies the symmetry-gap theorem of
\cref{sec:symmetry-gap} to this input.

\subsubsection{Graph notation}
A \intro{multigraph} is a pair $G=(V(G),E(G))$, where $V(G)$ is a finite set of \intro{vertices} and $E(G)$ is a finite multiset of \intro{edges}, which are elements of $\binom{V(G)}{2}$. Throughout, we say \intro{graph} for multigraph. The \intro{degree} of a vertex $v$ is the number of edges that have $v$ as an endpoint. The graph is
\intro[d-regular]{$d$-regular} if all degrees equal $d$. The \intro{neighborhood} map
$\intro*\Nbr\colon2^{V(G)}\to2^{V(G)}$
sends $X$ to the set $\Nbr(X)$ of vertices adjacent to some $x\in X$.

A graph is called \intro{bipartite} if its vertex set can be partitioned into a \intro{left side} $\mathcal U$ and a \intro{right side} $\mathcal V$ such that each edge contains exactly one left and one right endpoint.
For a finite \intro{bipartite graph} $\intro*\BipGraph$ with left side
$\mathcal U$ and right side $\mathcal V$,
the \intro{cover number} $\intro*\CoverNum\colon2^{\mathcal U}\to\mathbb Z_{\geq0}$ is
\[
 \CoverNum(U)\coloneqq
 \min\{|V|:V\subseteq\mathcal V,\ U\subseteq\Nbr(V)\}.
\]

\subsubsection{A graph with expansion and a large cover number}
\label{sec:quant-source}

Fix a sufficiently large integer $d$ and set
\begin{equation}\label{eq:quant-parameters}
 K\coloneqq d^{d+3},\qquad m\coloneqq K^2.
\end{equation}
In the following, we prove the existence of a $d$-regular bipartite graph with $m$ vertices on each side
that has the large cover number and the expansion described at the
beginning of this subsection. Graphs with these properties can also be
obtained from existing results, for example from the regular graphs of
large girth and large domination number due to Alon~\cite[Lemma~2.1(ii)]{Alon2010}. Girth larger than $4K$ gives the
expansion, but the girth of these graphs is only logarithmic in their
number $n$ of vertices, so this requires $n\geq d^{8K}$. The resulting
gap of order $\log d$ is then only of order $\log\log\log|\mathcal C|$.
Moreover, these results are stated for fixed $d$ as $n\to\infty$, and
making their dependence on $d$ explicit requires going through their
proofs. The following lemma shows that a random graph has these
properties, which gives a gap of order $\log\log|\mathcal C|$.

\begin{lemma}[Random base graph]\label{lem:quant-source}
There is a \kl[d-regular]{$d$-regular} bipartite multigraph $\BipGraph$
with sides $\mathcal U,\mathcal V$ whose edges are the union of $d$
perfect matchings $M_a$, $a\in\mathbb Z_d$, such that
\[
 |\mathcal U|=|\mathcal V|=m,\qquad
 \CoverNum(\mathcal U)\geq\frac{m\log d}{2d},
\]
and
\begin{equation}\label{eq:base-expansion}
 |\Nbr(V)|\geq(d-2)|V|
 \qquad \forall \; V\subseteq\mathcal V,\ |V|\leq2K.
\end{equation}
\end{lemma}
Small-set expansion in the union of $d$ independent uniform perfect
matchings is well known~\cite[Theorem~4.4 and subsequent discussion]{Vadhan2012}. For completeness, we provide a self-contained proof and verify the required large cover size.

\begin{proof}
Take the union of $d$ perfect matchings between
$\mathcal U$ and $\mathcal V$, chosen independently and uniformly at random, and color each edge by the index of
its matching. This gives a $d$-regular bipartite multigraph. We show that
each required property fails with probability $o(1)$ as $d\to\infty$.

We start with the cover number. Let $r\coloneqq\lfloor m\log d/(2d)\rfloor$ and fix an $r$-set
$V\subseteq\mathcal V$. Its images $S_1,\ldots,S_d$ under the
matchings are independent uniform $r$-subsets of $\mathcal U$.
Write $S\coloneqq\bigcup_a S_a$ and $q\coloneqq(1-r/m)^d =\Pr[u\notin S]$, the probability that a
given left vertex $u$ is uncovered.

We will see in the following, that the events that distinct left vertices are covered are negatively correlated, which will allow us to upper bound the probability that $\mathcal U$ is covered by $(1-q)^m$.
 More precisely, let $B\subseteq \mathcal U$ and let $u\in \mathcal U\setminus B$. 
We show that $\Pr[B\subseteq S \mid u\notin S]\ge \Pr[B\subseteq S]$, i.e., the probability of $B$ being covered by $S$, conditioned on $u$ remaining uncovered, is at least as large as the (unconditional) probability of $B$ being covered.
To see this, we observe that under the event that $u$ is uncovered, each $S_a$ is an $r$-subset of $\mathcal U\setminus\{u\}$ and by symmetry, each choice in $\binom{\mathcal U\setminus\{u\}}{r}$ has the same probability. We can recover this probability distribution by the following two-step process: First, select matchings as before, allowing $u$ as a left endpoint. Next, whenever $u$ is selected by a matching, replace it by a left vertex uncovered by that matching, selected uniformly at random. By symmetry, the sets of left endpoints are distributed uniformly in $\binom{\mathcal U\setminus\{u\}}{r}$. Moreover, whenever the original sets of endpoints cover $B$, so do the new sets of endpoints.
This shows that $\Pr[B\subseteq S \mid u\notin S]\ge \Pr[B\subseteq S]$. Using the identity $\Pr[B\subseteq S]=\Pr[u\in S] \cdot \Pr[B\subseteq S \mid u\in S]+(1-\Pr[u\in S])\cdot \Pr[B\subseteq S \mid u\notin S]$, we can infer that
$\Pr[B\subseteq S\mid u\in S]\le \Pr[B\subseteq S]$. In particular,
\[\Pr[B\cup \{u\}\subseteq S]=\Pr[B\subseteq S\mid u\in S]\cdot \Pr[u\in S]\le \Pr[B\subseteq S]\cdot \Pr[u\in S].\]
Let $\mathcal U = \{u_1,\dots,u_m\}$ and apply the previous inequality inductively with $u=u_i$ and $B=\{u_1,\dots,u_{i-1}\}$. This yields 
\begin{equation}\label{eq:quant-source-one-set}
 \Pr[S=\mathcal U]\le \prod_{i=1}^m \Pr[u_i\in S]=(1-q)^m\leq e^{-mq}.
\end{equation}
Since $r\leq m\log d/(2d)$, with $n\coloneqq2d/\log d$ we have
\[
 q=\Bigl(1-\frac rm\Bigr)^d
 \geq\Bigl(1-\frac1n\Bigr)^{n\log d/2}
 \geq\Bigl(1-\frac1n\Bigr)^{\log d/2}\frac1{\sqrt d}
 =\frac{1-o(1)}{\sqrt d},
\]
for sufficiently large $d$, where the second inequality uses $(1-1/n)^{n}\geq e^{-1}(1-1/n)$. Furthermore, by the standard bound we get 
$\log\binom mr\leq r\log(em/r)=O(m\log^2d/d)=o(m/\sqrt d)$. A union
bound over the $r$-sets $V$ together
with~\eqref{eq:quant-source-one-set} gives, for sufficiently large $d$,
\begin{equation}\label{eq:quant-source-cover-probability}
 \Pr[\text{some }r\text{-set covers }\mathcal U]
 \leq\binom mr e^{-mq}
 \leq\exp\!\left(-(1-o(1))\frac{m}{\sqrt d}\right)\leq\frac1d.
\end{equation}
If no $r$-set covers $\mathcal U$, no smaller set does either, so
$\CoverNum(\mathcal U)\geq r+1>m\log d/(2d)$.

For the expansion bound~\eqref{eq:base-expansion}, fix $1\leq k\leq2K$ and put $\ell\coloneqq(d-2)k<m$.
If a $k$-set $V$ has at most $\ell$ neighbors, they lie in some
$\ell$-set $U$. A uniform matching sends all of $V$ into $U$ with
probability
\begin{equation}\label{eq:quant-source-one-matching}
 \frac{\binom\ell k}{\binom mk}
 =\prod_{i=0}^{k-1}\frac{\ell-i}{m-i}
 \leq\left(\frac\ell m\right)^k,
\end{equation}
because $(\ell-i)/(m-i)\leq\ell/m$ when $\ell\leq m$. In particular, all $d$ matchings send $V$ into $U$ with probability at most $\left(\frac\ell m\right)^{dk}$.
A union bound over $V$ and $U$,
with~\eqref{eq:quant-source-one-matching} for each of the $d$
independent matchings, followed by $\binom ab\leq(ea/b)^b$,
$dk-\ell=2k$ and $k\le 2K$, $m=K^2$ gives
\begin{equation}\label{eq:quant-source-expansion-k}
 \begin{aligned}
 \Pr[\exists V:\ |V|=k,\ |\Nbr(V)|\leq(d-2)k]
 &\leq\binom mk\binom m\ell\left(\frac\ell m\right)^{dk}\\
 &\leq\left(\frac{em}{k}\right)^k
       \left(\frac{em}{\ell}\right)^\ell
       \left(\frac\ell m\right)^{dk}\\
 &=\left[\frac{em}{k}\,e^{d-2}
       \left(\frac\ell m\right)^2\right]^k
 =\left[\frac{e^{d-1}(d-2)^2k}{m}\right]^k\\
 &\le \left[\frac{2e^{d-1}(d-2)^2}{K}\right]^k
 \end{aligned}
\end{equation}
Let
$\theta\coloneqq2e^{d-1}(d-2)^2/K$, then $\theta=o(1)$ because $K=d^{d+3}$.
Summing~\eqref{eq:quant-source-expansion-k} over $k$ therefore bounds
the failure probability by
\[
 \sum_{k=1}^{2K}\left[\frac{e^{d-1}(d-2)^2k}{m}\right]^k
 \leq\sum_{k=1}^{2K}\theta^k\leq\frac\theta{1-\theta}=o(1).
\]
The empty set satisfies the expansion bound as well. Both properties
therefore hold simultaneously with probability $1-o(1)>0$.
\end{proof}

\subsubsection{The cyclic lift}\label{sec:quant-seed}

\begin{definition}[Cyclic lift]\label{def:cyclic-lift}
Call the graph $\BipGraph$ of \cref{lem:quant-source} the
\intro{base graph}. Its vertices and edges are denoted as base vertices and base
edges. We define the \intro{cyclic lift}, a \kl{bipartite graph} $\widetilde G$ with sides
$\widetilde{\mathcal U}\coloneqq\mathcal U\times\mathbb Z_d$ and 
$\widetilde{\mathcal V}\coloneqq\mathcal V\times\mathbb Z_d$, by lifting each base edge $\{u,v\}\in M_a$ to the $d$ edges
\[
 \{(u,t+a),(v,t)\},\qquad t\in\mathbb Z_d.
\]
Note that each base edge
$\{u,v\}$ lifts to a perfect matching between $\{u\}\times\mathbb Z_d$ and
$\{v\}\times\mathbb Z_d$.

We call the $dm$ elements $\tilde u=(u,t)$ of $\widetilde{\mathcal U}$
\intro[lifted client]{lifted clients} and the elements $\tilde v=(v,t)$ of
$\widetilde{\mathcal V}$ \intro[lifted right vertex]{lifted right vertices}. We point out that the \kl{lifted right vertices} are only used in order to define the facility cost function; the facilities themselves will correspond to the base right vertices. Furthermore, we define $t$ as their \intro{cyclic coordinate}. \cref{sec:quant-hiding}
replaces each \kl{lifted client} by $R$ copies. The \intro{fiber} over a base
vertex $u$ is $\{u\}\times\mathbb Z_d$, and for $U\subseteq\mathcal U$
the union of the fibers over $U$ is $U\times\mathbb Z_d$; we call such
sets unions of \intro[whole fibers]{whole fibers}. Likewise, the fiber
$\{v\}\times\mathbb Z_d$ over $v\in\mathcal V$ consists of the
\intro[lift]{lifts} of $v$.
\end{definition}
\begin{figure}[!t]
\centering
\begin{tikzpicture}[
  x=1cm,y=1cm,
  fac/.style={draw,thick,rectangle,rounded corners=1pt,fill=white,
    minimum size=5.5mm,inner sep=0pt,font=\small},
  lft/.style={draw,thick,circle,fill=white,minimum size=5.5mm,
    inner sep=0pt,font=\small},
  lv/.style={draw,circle,minimum size=4.4mm,inner sep=0pt,font=\scriptsize},
  rv/.style={draw,rectangle,rounded corners=1pt,minimum size=4.4mm,
    inner sep=0pt,font=\scriptsize},
  cp/.style={circle,draw=black!70,line width=0.3pt,minimum size=1.35mm,
    inner sep=0pt},
  fiber/.style={draw=black!45,densely dashed,rounded corners=4pt},
  be/.style={line width=1.2pt},
  e/.style={line width=0.5pt},
  ed/.style={line width=1.3pt},
  stub/.style={thick,black!35,densely dotted},
  blab/.style={font=\scriptsize,fill=white,inner sep=1pt},
  flab/.style={font=\scriptsize,text=black!70},
  tier/.style={font=\small,anchor=east,align=right},
  op/.style={-stealth,thick,black!55},
]
\colorlet{t0}{cyc0}\colorlet{tt0}{black}
\colorlet{t1}{cyc1}\colorlet{tt1}{black}
\colorlet{t2}{cyc2}\colorlet{tt2}{white}
\colorlet{a0}{definitioncolor}
\definecolor{a1}{HTML}{1F5FA8}
\definecolor{a2}{HTML}{E0A100}
\def\W{2.6}\def\s{0.74}\def\g{0.16}\def\gh{0.24}
\def\baseedges{1/1/0,1/2/1,1/3/2,2/2/0,2/3/1,2/4/2}
\def\yV{0}\def\yU{-1.5}
\def\yLV{-3.35}\def\yLU{-4.85}
\def\yCV{-6.7}\def\yCU{-8.3}
\foreach \j/\i/\a in \baseedges
  \draw[be,draw=a\a] ({\j*\W},\yV)--({(\i-1)*\W},\yU)
    node[blab,pos=.34,text=a\a] {$\a$};
\foreach \i/\dir in {1/-1,1/-0.5,2/1,3/-1,4/1,4/0.5}
  \draw[stub] ({(\i-1)*\W},\yU)--++({0.45*\dir},0.75);
\foreach \j in {1,2} \node[fac] at ({\j*\W},\yV) {$v_\j$};
\foreach \i in {1,...,4} \node[lft] at ({(\i-1)*\W},\yU) {$u_\i$};
\foreach \j/\i/\a in \baseedges
  \foreach \t in {0,1,2}{
    \pgfmathtruncatemacro\tp{mod(\t+\a,3)}
    \ifnum\t=0 \def\est{ed}\def\ecol{a\a}\else\def\est{e}\def\ecol{a\a!75}\fi
    \draw[\est,draw=\ecol] ({\j*\W+(\t-1)*\s},\yLV)--({(\i-1)*\W+(\tp-1)*\s},\yLU);
  }
\foreach \j in {1,2}{
  \draw[fiber] ({\j*\W-\s-0.38},\yLV-0.32) rectangle
    ({\j*\W+\s+0.38},\yLV+0.32);
  \node[flab,anchor=south] at ({\j*\W},\yLV+0.34)
    {$\{v_\j\}\times\mathbb Z_d$};
  \foreach \t in {0,1,2}
    \node[rv,fill=t\t,text=tt\t] at ({\j*\W+(\t-1)*\s},\yLV) {$\t$};
}
\foreach \i in {1,...,4}{
  \draw[fiber] ({(\i-1)*\W-\s-0.38},\yLU-0.32) rectangle
    ({(\i-1)*\W+\s+0.38},\yLU+0.32);
  \node[flab,anchor=north] at ({(\i-1)*\W},\yLU-0.34)
    {$\{u_\i\}\times\mathbb Z_d$};
  \foreach \t in {0,1,2}
    \node[lv,fill=t\t,text=tt\t] at ({(\i-1)*\W+(\t-1)*\s},\yLU) {$\t$};
}
\foreach \j/\i/\a in \baseedges
  \foreach \t in {0,1,2}{
    \pgfmathtruncatemacro\tp{mod(\t+\a,3)}
    \ifnum\t=0 \def\est{ed}\def\ecol{a\a}\else\def\est{e}\def\ecol{a\a!75}\fi
    \draw[\est,draw=\ecol] ({\j*\W+(\t-1)*\s},\yCV)--({(\i-1)*\W+(\tp-1)*\s},\yCU+\gh+0.1);
  }
\foreach \j in {1,2}
  \foreach \t in {0,1,2}
    \node[rv,fill=t\t,text=tt\t] at ({\j*\W+(\t-1)*\s},\yCV) {$\t$};
\foreach \i in {1,...,4}{
  \draw[fiber] ({(\i-1)*\W-\s-0.36},\yCU-\gh-0.13) rectangle
    ({(\i-1)*\W+\s+0.36},\yCU+\gh+0.13);
  \node[flab,anchor=north] at ({(\i-1)*\W},\yCU-\gh-0.15)
    {public block $\AbstractBlock(u_\i)$};
  \foreach \t in {0,1,2}{
    \pgfmathsetmacro\cx{(\i-1)*\W+(\t-1)*\s}
    \draw[black!55,line width=0.4pt,rounded corners=1.5pt,fill=black!4]
      ({\cx-\gh-0.1},\yCU-\gh-0.1) rectangle ({\cx+\gh+0.1},\yCU+\gh+0.1);
    \foreach \p in {0,...,3}
      \foreach \q in {0,...,3}
        \node[cp,fill=t\t] at ({\cx+(\p-1.5)*\g},{\yCU+(\q-1.5)*\g}) {};
  }
}
\def\lx{-1.4}
\node[tier] at (\lx,{(\yV+\yU)/2})
  {(a) base graph\\$\BipGraph$ on $\mathcal U\cup\mathcal V$};
\node[tier] at (\lx,{(\yLV+\yLU)/2})
  {(b) cyclic lift\\$\widetilde G$ on $\widetilde{\mathcal U}\cup\widetilde{\mathcal V}$};
\node[tier] at (\lx,{(\yCV+\yCU)/2})
  {(c) copy graph\\on $\mathcal C\cup\widetilde{\mathcal V}$};
\def\ax{-2.3}\def\px{3*\W+1.3}
\foreach \ya/\yb/\down/\up in {
    \yU-0.4/\yLV+0.75/{lift $\times\mathbb Z_d$}/{$(x,t)\mapsto\AbstractOrbitIndex(x,t)$},
    \yLU-0.75/\yCV+0.4/{copy $\times R$}/{$c\mapsto\AbstractLabelMap(c)$}}{
  \draw[op] (\ax,\ya)--(\ax,\yb)
    node[midway,left,font=\scriptsize,text=black!75] {\down};
  \draw[op] ({\px},\yb)--({\px},\ya)
    node[midway,right,font=\scriptsize,text=black!75] {\up};
}
\end{tikzpicture}
\caption{The \kl{base graph}, its \kl{cyclic lift}, and the copy graph on
a fragment, drawn with $d=3$ and $R=16$. Squares are right vertices and
circles are left vertices. The columns are aligned: everything over a
base vertex $u_i$ is drawn directly below it.
(a)~Each base edge is labeled and colored by the index $a$ of its
matching $M_a$; dotted stubs lead to the rest of $\BipGraph$.
(b)~Each base vertex becomes its \kl{fiber} (dashed box), whose vertices
are labeled and shaded by their \kl{cyclic coordinate} $t$. A base edge
$\{u,v\}\in M_a$ lifts to the perfect matching $\{(u,t+a),(v,t)\}$,
$t\in\mathbb Z_d$, between the two fibers, drawn in the color of $\{u,v\}$:
red edges keep $t$, and blue and yellow edges shift it by one and two,
respectively.
The bold edges from $(v_1,0)$
and $(v_2,0)$ join them to the \kl[first copy neighborhood]{first copy neighborhoods}
$\DiagSet_{v_1}$ and $\DiagSet_{v_2}$, each of which takes one
\kl{lifted client} from each neighboring fiber.
(c)~Each \kl{lifted client} $(u,t)$ becomes its \kl[abstract copy class]{copy class}
$\AbstractLabelMap^{-1}(u,t)$ of $R$ clients (solid box), and each \kl[abstract copy]{copy} is joined
to every neighbor of its \kl[abstract label]{label} in $\widetilde G$; an edge drawn to a
solid box stands for the $R$ edges to its copies. The fiber over $u$
becomes the \kl[abstract public block]{public block} $\AbstractBlock(u)$ of $dR=48$ clients located at
$u$. The algorithm sees the dashed blocks; the solid boxes and the
shading show the \kl[abstract hidden partition]{hidden partition}. The maps on the right
send each graph to the one above it and map edges to edges.}
\label{fig:three-graphs}
\end{figure}
\cref{fig:three-graphs} shows the base graph, its cyclic lift, and the
replicated copy graph of \cref{sec:quant-hiding}.
The cyclic lift is a well-known object and is the derived graph of a voltage graph in the sense
of Gross~\cite{Gross1974}. 
Lifts of graphs have been studied extensively,
see~\cite{AmitLinial2002} for random lifts and~\cite{AgarwalEtAl2019} for cyclic lifts and their expansion.
The neighborhood map of $\widetilde G$ on right sets is
$\intro*\LiftNbr\colon2^{\widetilde{\mathcal V}}\to2^{\widetilde{\mathcal U}}$,
$\tilde V\mapsto\LiftNbr(\tilde V)$; in particular,
\[
 \LiftNbr(\{(v,t)\})=\{(u,t+a):a\in\mathbb Z_d,\ \{u,v\}\in M_a\}.
\]
The \intro[first copy neighborhood]{first copy neighborhoods} 
$\intro*\DiagSet_v\coloneqq\LiftNbr(\{(v,0)\})$,
 $v\in\mathcal V$, partition
$\widetilde{\mathcal U}$: the \kl{lifted client} $(u,a)$ lies in
$\DiagSet_v$ exactly when $v$ is matched to $u$ in $M_a$.
For $a\in\mathbb Z_d$, the \intro{shift} by $a$ is the bijection
$\gamma_a\colon(x,t)\mapsto(x,t+a)$ of $\widetilde{\mathcal U}$ and of
$\widetilde{\mathcal V}$. We write $\tilde U+a\coloneqq\gamma_a(\tilde U)$
and $\tilde V+a\coloneqq\gamma_a(\tilde V)$ for
$\tilde U\subseteq\widetilde{\mathcal U}$ and
$\tilde V\subseteq\widetilde{\mathcal V}$. The shifts form a
\kl{permutation group} of graph automorphisms of $\widetilde G$. Indeed, every shift maps edges of $\widetilde G$ to edges,
since it sends $\{(u,t+a'),(v,t)\}$ to $\{(u,t+a+a'),(v,t+a)\}$. Hence
$\LiftNbr(\tilde V+a)=\LiftNbr(\tilde V)+a$.
The orbits of the shifts on $\widetilde{\mathcal U}$ are the fibers,
so the subsets of $\widetilde{\mathcal U}$ that are \kl{invariant} under
the shifts are exactly the unions of \kl{whole fibers}.

\subsubsection{The sparse rank function}\label{sec:quant-cap}

In this section, we discuss how to construct the facility cost function. Before giving a completely formal definition, we provide some high-level intuition.
The facility cost function is designed so that the neighborhood of every
\kl{lifted right vertex} is cheap to serve, while a union of \kl{whole fibers} is
not. A \kl{lifted client} costs $1$ on its own. A lifted right vertex
$\tilde v$ can act as a witness that covers all $d$ clients in its
neighborhood at a total cost of only $3$; this is the discount of $d-3$ per
witness in~\eqref{eq:quant-service-rank} below. In particular, the
\kl[first copy neighborhood]{first copy neighborhoods} $\DiagSet_v=\LiftNbr(\{(v,0)\})$ have cost at most
$3$. If the neighborhoods of the witnesses overlap heavily, the
witnesses receive their discount without covering new clients. Local
expansion rules this out
for small witness sets: they have at least $d-2$ neighbors per
witness, more than the discount. We therefore allow witness sets of size at most $2K$ and cap the function value at $K$. The cap forces a minimizing witness
set of value below $K$ to have fewer than $K$ elements, which keeps
the function submodular (\cref{lem:quant-cap}).

\begin{definition}[Sparse rank function]\label{def:sparse-rank}
A \kl{witness} $\tilde V\subseteq\widetilde{\mathcal V}$ is
\intro{admissible} if $|\tilde V|\leq2K$. We say that $\widetilde G$ has
\intro[local expansion]{local expansion up to $2K$} if it satisfies
the analogue of~\eqref{eq:base-expansion}:
\begin{equation}\label{eq:quant-cost-vs-size}
 |\LiftNbr(\tilde V)|\geq(d-2)|\tilde V|
 \qquad\text{for every \kl{admissible} }\tilde V.
\end{equation}
The \intro{sparse rank function}
$\intro*\SparseCost\colon2^{\widetilde{\mathcal U}}\to\mathbb Z_{\geq0}$ with
\intro{cap} $K$ and its witness cost
$\intro*\SparseObj\colon2^{\widetilde{\mathcal U}}\times2^{\widetilde{\mathcal V}}\to\mathbb Z$
are
\begin{equation}\label{eq:quant-service-rank}
 \begin{aligned}
 \SparseCost(\tilde U)
 &\coloneqq\min\Bigl\{K,\ \min_{\tilde V\text{ \kl{admissible}}}\SparseObj(\tilde U,\tilde V)\Bigr\},\\
 \SparseObj(\tilde U,\tilde V)
 &\coloneqq|\tilde U\cup\LiftNbr(\tilde V)|-(d-3)|\tilde V|
 \end{aligned}
\end{equation}
\end{definition}

The next lemma shows that the cyclic lift inherits local expansion from
the base graph, and that the rank of a union of whole fibers is at
least $d$ times the cover number of the corresponding base vertices, up
to the cap.

\begin{lemma}[Expansion and the rank of unions of fibers]
\label{lem:lift-expansion}\label{lem:quant-fiber-cover}
The \kl{cyclic lift} has \kl{local expansion} up to $2K$, and for every
$U\subseteq\mathcal U$,
\begin{equation}\label{eq:quant-fiber-cover}
 \SparseCost(U\times\mathbb Z_d)
 \geq\min\{K,\ d\,\CoverNum(U)\}.
\end{equation}
\end{lemma}
\begin{proof}
Let $\tilde V$ be admissible, let $k_v\coloneqq|\tilde V\cap(\{v\}\times\mathbb Z_d)|$
count its \kl[lift]{lifts} of $v$, and let $V_j\coloneqq\{v:k_v\geq j\}$,
$j=1,\ldots,d$, be the level sets.
Then $\sum_j|V_j|=\sum_vk_v=|\tilde V|$.
A base edge $\{u,v\}$ lifts to a perfect matching between the fibers over
$u$ and $v$, so the $k_v$ selected \kl[lift]{lifts} of $v$ have $k_v$ distinct
neighbors in the fiber over $u$. Neighbors coming from different $v$
may coincide, so we count only the largest contribution. Hence
\begin{equation}\label{eq:lift-level-count}
 \begin{aligned}
 |(U\times\mathbb Z_d)\cup\LiftNbr(\tilde V)|
 &\geq d|U|+\sum_{u\notin U}\max_{v\in\Nbr(\{u\})}k_v\\
 &=\sum_{j=1}^d|U\cup\Nbr(V_j)|.
 \end{aligned}
\end{equation}
The equality holds because, for $u\notin U$, the maximum
$\max_{v\in\Nbr(\{u\})}k_v$ is the number of indices $j$ (or, equivalently, the maximum index $j$) with
$u\in\Nbr(V_j)$.
The expansion property in the base graph~\eqref{eq:base-expansion} applies to each $V_j$ because
$|V_j|\leq|\tilde V|\leq2K$. Hence
\begin{align*}
 \SparseObj(U\times\mathbb Z_d,\tilde V)
 &=|(U\times\mathbb Z_d)\cup\LiftNbr(\tilde V)|-(d-3)|\tilde V|
   \tag*{(by~\eqref{eq:quant-service-rank})}\\
 &\geq\sum_{j=1}^d
   \bigl(|U\cup\Nbr(V_j)|-(d-3)|V_j|\bigr)
   \tag*{(by~\eqref{eq:lift-level-count} and $\textstyle\sum_j|V_j|=|\tilde V|$)}\\
 &=\sum_{j=1}^d
   \bigl(|U\setminus \Nbr(V_j)| + |\Nbr(V_j)|-(d-2)|V_j| + |V_j|\bigr)
   \tag*{(splitting $U\cup\Nbr(V_j)$)}\\
 &\geq\sum_{j=1}^d\bigl(|U\setminus\Nbr(V_j)|+|V_j|\bigr)
   \tag*{(by~\eqref{eq:base-expansion})}\\
 &\geq d\,\CoverNum(U).
   \tag*{(by covering $U$ from $V_j$)}
\end{align*}
In the last step, adding to $V_j$ one neighbor of each vertex of
$U\setminus\Nbr(V_j)$ gives a cover of $U$, so each summand is at least
$\CoverNum(U)$. Minimizing over admissible witnesses and applying the
cap proves the lemma. For $U=\varnothing$, the second-to-last line
gives $|\LiftNbr(\tilde V)|-(d-3)|\tilde V|\geq|\tilde V|$, which is
local expansion~\eqref{eq:quant-cost-vs-size}.
\end{proof}

We now show that $\SparseCost$ is a matroid rank function that is small on the
neighborhood of every lifted right vertex and \kl{invariant} under shifts.

\begin{lemma}[Properties of the sparse rank function]\label{lem:quant-cap}
The function $\SparseCost$ is a
\kl[rank function]{matroid rank function} with
$\SparseCost(\LiftNbr(\{\tilde v\}))\leq3$ for every
$\tilde v\in\widetilde{\mathcal V}$, and it is
\kl{invariant} under shifts, i.e., 
$\SparseCost(\tilde U+a)=\SparseCost(\tilde U)$ for all $\tilde U\subseteq\widetilde{\mathcal U}$ and $a\in\mathbb Z_d$.
\end{lemma}
\begin{proof}
The function is integer-valued and monotone. The empty witness gives
$0\leq\SparseCost(\tilde U)\leq|\tilde U|$. For every admissible
$\tilde V$, local expansion~\eqref{eq:quant-cost-vs-size}
(\cref{lem:lift-expansion}) gives
\begin{equation}\label{eq:quant-witness-size}
 \SparseObj(\tilde U,\tilde V)
 =|\tilde U\cup\LiftNbr(\tilde V)|-(d-3)|\tilde V|
 \geq|\LiftNbr(\tilde V)|-(d-3)|\tilde V|\geq|\tilde V|.
\end{equation}

Viewed on the combined \kl{ground set}
$\widetilde{\mathcal U}\sqcup\widetilde{\mathcal V}$, $\SparseObj$ is
submodular, since by~\eqref{eq:quant-service-rank} it is a coverage
function to which a modular function is added.
To show submodularity of the function $\SparseCost$, we would like to apply \itemref{lem:shared-minimization}{item:min-partial}. However, we need to be careful because we are restricting the minimum to admissible sets and the union of two admissible sets is not always admissible.
Let
$\tilde U_1,\tilde U_2\subseteq\widetilde{\mathcal U}$. If
$\SparseCost(\tilde U_1)=K$, then also $\SparseCost(\tilde U_1 \cup \tilde U_2)=K$  by monotonicity and the \kl{cap}. Hence, monotonicity yields
\[
 \SparseCost(\tilde U_1\cup \tilde U_2)
 +\SparseCost(\tilde U_1\cap \tilde U_2)
 \leq K+\SparseCost(\tilde U_2)
 =\SparseCost(\tilde U_1)+\SparseCost(\tilde U_2),
\]
and the case $\SparseCost(\tilde U_2)=K$ is symmetric. Otherwise
$\SparseCost(\tilde U_i)<K$, so $\SparseCost(\tilde U_i)$ is attained by
an admissible \kl{minimizing witness} $\tilde V_i$, and
by~\eqref{eq:quant-witness-size}
\[
 |\tilde V_i|\leq\SparseObj(\tilde U_i,\tilde V_i)=\SparseCost(\tilde U_i)<K.
\]
The union and intersection of $\tilde V_1$ and $\tilde V_2$ are therefore admissible, so the proof of
\itemref{lem:shared-minimization}{item:min-partial} applies verbatim: submodularity
of $\SparseObj$ at $(\tilde U_1,\tilde V_1),(\tilde U_2,\tilde V_2)$
gives submodularity of $\SparseCost$ at $\tilde U_1,\tilde U_2$.
The rank criterion of \cref{sec:statement} makes $\SparseCost$
a matroid rank function. 

To bound $\SparseCost(\LiftNbr(\{\tilde v\}))$, consider the witness $\{\tilde v\}$, which gives a cost of
$|\LiftNbr(\{\tilde v\})|-(d-3)\leq3$. 

Finally, shifting by $a$ preserves
admissibility, and $\SparseObj(\tilde U+a,\tilde V+a)=\SparseObj(\tilde U,\tilde V)$,
so both minima defining $\SparseCost$ are shift \kl{invariant}.
\end{proof}

\subsubsection{The lifted instance and its symmetry gap}
\label{sec:metric-definition}

\begin{definition}[Lifted instance]\label{def:lifted-instance}
Define the metric $\dist\colon(\mathcal U\sqcup\mathcal V)^2\to[0,1]$ on the
base vertices by
\[
 \dist(x,y)\coloneqq\min\{1,\ h(x,y)/d\},
\]
where $h(x,y)$ is the number of edges of a shortest path from $x$ to $y$ in
$\BipGraph$, and $h(x,y)=+\infty$ if there is none.
 Truncation preserves the \intro{triangle inequality}, so $\dist$ is a metric.
We remark that the main reason we truncate is to guarantee finite distances.
Placing each \kl{lifted client} $(u,t)$ at the point $u$ gives the
\intro{lifted instance} $(\widetilde{\mathcal U},\mathcal V,\dist,\SparseCost)$.
Its facilities are the base right vertices $\mathcal V$; the
\kl[lifted right vertex]{lifted right vertices} appear only as witnesses in the definition of
$\SparseCost$. \cref{sec:quant-hiding} replicates its clients.

The shift orbits are precisely the \kl[fiber]{fibers}, so
\kl{whole-orbit} assignments keep each fiber together. The \intro{first copy assignment} sends $\DiagSet_v$ to $v$ for
every $v\in\mathcal V$; it is an assignment because the
\kl[first copy neighborhood]{first copy neighborhoods} partition $\widetilde{\mathcal U}$.
\end{definition}
The first copy assignment splits the fiber over $u$ among the neighbors
of $u$: it sends $(u,a)$ to the vertex matched to $u$ in $M_a$. The
next lemma shows that the first copy assignment is cheap and that every
\kl{whole-orbit} assignment is expensive.

\begin{lemma}[First copy and whole-orbit assignments]\label{lem:symmetry-gap}
\leavevmode
\begin{enumerate}
\item\label{item:diagonal-cost} The \kl{first copy assignment} has cost
at most $4m$.
\item\label{item:whole-fiber-cost} Let $(U_v)_{v\in \mathcal V}$ be subsets of $\mathcal U$
satisfying $\dist(u,v)<1$ for every $u\in U_v$, and put
$U'\coloneqq\bigcup_vU_v$. Then
\begin{equation}\label{eq:whole-fiber-cover}
 \sum_{v\in\mathcal V}\SparseCost(U_v\times\mathbb Z_d)
 +d\,|\mathcal U\setminus U'|
 \geq d\,\CoverNum(\mathcal U).
\end{equation}
In particular, every \kl{whole-orbit} assignment has cost at least
$d\,\CoverNum(\mathcal U)$.
\end{enumerate}
\end{lemma}
\begin{proof}
For \cref{item:diagonal-cost}, note that each cluster $\DiagSet_v=\LiftNbr(\{(v,0)\})$ has rank at most three by
\cref{lem:quant-cap}, and each of its clients $(u,a)$ lies over a
neighbor $u$ of $v$, at distance $1/d$. The clusters partition the $dm$
lifted clients, so the first copy assignment costs
\[
 \sum_{v\in\mathcal V}\Bigl(\SparseCost(\DiagSet_v)
   +\sum_{(u,a)\in\DiagSet_v}\dist(u,v)\Bigr)
 \leq3m+\frac{dm}{d}=4m.
\]

For \cref{item:whole-fiber-cost}, fix $v$. Every $u\in U_v$ is joined to $v$ by a path of at most $d-1$
edges, so $U_v$ lies in the radius-$(d-1)$ ball around $v$, which has at
most $\sum_{j=0}^{d-1}d^j\leq d^d$ vertices. Consequently
$d\,\CoverNum(U_v)\leq d|U_v|\leq d^{d+1}\leq K$,
and~\eqref{eq:quant-fiber-cover} gives
$\SparseCost(U_v\times\mathbb Z_d) \geq \min\{K, d\,\CoverNum(U_v)\} = d\,\CoverNum(U_v)$. Covers of the
sets $U_v$, together with one neighbor of each vertex of
$\mathcal U\setminus U'$, cover $\mathcal U$, so
$\sum_v\CoverNum(U_v)+|\mathcal U\setminus U'|\geq\CoverNum(\mathcal U)$.
Multiplying by $d$ proves~\eqref{eq:whole-fiber-cover}.

Now let a \kl{whole-orbit} assignment send the fibers over
$W_v\subseteq\mathcal U$ to $v$, and put
$U_v\coloneqq\{u\in W_v:\dist(u,v)<1\}$. The fiber over a vertex of
$\mathcal U\setminus U'$ goes to a facility at distance one, so its $d$
clients pay connection cost $d$, and
$\SparseCost(W_v\times\mathbb Z_d)\geq\SparseCost(U_v\times\mathbb Z_d)$
by monotonicity. The cost of the assignment is therefore at least the
left-hand side of~\eqref{eq:whole-fiber-cover}.
\end{proof}

\subsection{Applying the symmetry-gap theorem to SFL}\label{sec:quant-hiding}

In this subsection, we apply the symmetry-gap theorem of \cref{sec:symmetry-gap} to the
lifted instance of \cref{sec:quantitative-seed} and prove the lower bound for submodular facility location
(\cref{thm:quantitative}).

\subsubsection{Setting up the replication}\label{sec:sfl-abstract-application}

We replicate the input with
\[
 \mathcal C_0=\widetilde{\mathcal U},\qquad
 I=\mathcal U,\qquad \AbstractOrbit_u=\{u\}\times\mathbb Z_d,\qquad
 \mathcal F=\mathcal V,\qquad \AbstractInputCost=\SparseCost.
\]
The group $\AbstractGroup=\{\gamma_a:a\in\mathbb Z_d\}$ consists of the
\kl[shift]{shifts} $\gamma_a$ of
\cref{sec:quant-seed}, and its orbits are the fibers
$\AbstractOrbit_u$. The rank function
$\SparseCost$ is \kl{invariant} under these shifts by
\cref{lem:quant-cap}, and the connection cost to a facility
$v$ depends only on $u$. Moreover,
$\SparseCost(\AbstractOrbit_u)\leq |\AbstractOrbit_u|=d$ because $\SparseCost$ is a matroid
rank function. Thus, the lifted instance
$(\widetilde{\mathcal U},\mathcal V,\dist,\SparseCost,\AbstractGroup)$ is a \kl[abstract symmetric instance]{symmetric assignment instance}
satisfying~\eqref{item:fibered-single}, and its \kl{whole-orbit}
assignments are those of \cref{def:lifted-instance},
so \cref{lem:symmetry-gap} and the bound on $\CoverNum(\mathcal U)$ guaranteed by \cref{lem:quant-source} give
\begin{equation}\label{eq:sfl-input-gap}
 \AbstractInputOpt\leq4m,
 \qquad \text{and} \qquad
 \AbstractSymOpt\geq d\,\CoverNum(\mathcal U)
 \geq\frac{m\log d}{2}.
\end{equation}

It remains to choose the replication factor $R$ and the slack
$\nu$ in \cref{thm:symmetry-gap}.

\subsubsection{Proof of the SFL lower bound}\label{sec:quant-proof}

The slack $\nu$ in \cref{thm:symmetry-gap} should be at most of the order of the cost
$\AbstractInputOpt\leq4m$ of the cheap assignment, so that the gap
condition of the theorem (Equation ~\eqref{eq:symmetry-gap-condition}) still holds for $\alpha$ of order
$\log d$. At the same time, $R\nu^2/|\mathcal C_0|^2$ has to be large
enough to beat the union bound over the $d^{O(d)}$ clients and
facilities in \cref{thm:symmetry-gap}. We choose $\nu=m$ and $R=d^4$.
Since $|\mathcal C_0|=dm$, this gives $2R\nu^2/|\mathcal C_0|^2=2d^2$.
We first prove the following explicit form of \cref{thm:quantitative}.

\begin{proposition}[Explicit success-probability bound]\label{prop:quantitative-bound}
There are an absolute constant $\alpha_0>0$, a polynomial $P(\cdot)$, and infinitely many integers
$N$ with the following property. For every algorithm that makes at most
$Q$ value queries on every run, there is an instance with
$|\mathcal C|=N$, $|\mathcal F|\leq|\mathcal C|$, rational distances whose encoding length is bounded by $P(|\mathcal C|)$, and a matroid-rank
facility cost function, on which the algorithm returns an assignment of
cost at most $\alpha_0\log\log|\mathcal C|\cdot\OPT$ with probability at
most
\[
 \bigl(Q+|\mathcal C|\bigr)
 \exp\!\left(-\frac19\left(\frac{\log|\mathcal C|}{\log\log|\mathcal C|}\right)^{2}\right).
\]
\end{proposition}
\begin{proof}
For each sufficiently large $d$, choose
\begin{equation}\label{eq:direct-parameters}
 \nu\coloneqq m,\qquad R\coloneqq d^4,\qquad
 |\mathcal C|=dmR=d^5m,\qquad |\mathcal V|=m.
\end{equation}
Since $m=K^2=d^{2d+6}$ as chosen in \eqref{eq:quant-parameters}, we have $\log|\mathcal C|=(2d+11)\log d$. Hence
$\log\log|\mathcal C|=O(\log d)$, and for large $d$ also
$\log\log|\mathcal C|\geq\log d$ and
\begin{equation}\label{eq:direct-d-in-C}
 d\geq\frac{\log|\mathcal C|}{3\log\log|\mathcal C|}.
\end{equation}
Fix an absolute constant $\alpha_0>0$ small enough
that, for all sufficiently large $d$, the factor
$\alpha\coloneqq\alpha_0\log\log|\mathcal C|$ satisfies
\begin{equation}\label{eq:direct-gap-in-n}
 \alpha\bigl(\AbstractInputOpt+\nu\bigr)
 \leq5\alpha m
 <\frac{m\log d}{2}
 \leq\AbstractSymOpt,
\end{equation}
which is the gap condition~\eqref{eq:symmetry-gap-condition}.
The facility cost functions are matroid rank functions by
\cref{thm:symmetry-gap},
since $R\SparseCost$ is integer-valued and $R\nu/|\mathcal C_0|=d^3$ is an integer.
The distances are of the form $\frac{i}{d}$ with $i\in\{0,1,\dots,d\}$.

Let $\mathcal A$ be a randomized algorithm that makes at most $Q$
value queries on every run.
\cref{thm:symmetry-gap}, with $|\mathcal F|=m$,
$|\mathcal C_0|=dm$, and $2R\nu^2/|\mathcal C_0|^2=2d^2$, gives an
instance of the family on which $\mathcal A$ succeeds with probability at most
\begin{equation}\label{eq:direct-failure}
 \eta_d\coloneqq2dm\,(Q+m)e^{-2d^2}.
\end{equation}
Since $2dm=2d^{2d+7}$, we have $2dm\,e^{-d^2}\leq1$ for large $d$.
Together with $m\leq|\mathcal C|$, this gives the first inequality below,
and~\eqref{eq:direct-d-in-C} gives the second:
\[
 \eta_d\leq\bigl(Q+|\mathcal C|\bigr)e^{-d^2}
 \leq\bigl(Q+|\mathcal C|\bigr)
 \exp\!\left(-\frac19\left(\frac{\log|\mathcal C|}{\log\log|\mathcal C|}\right)^{2}\right).
\]
Finally, the client counts $N(d)=d^5m(d)$ are strictly increasing and
independent of the algorithm, and $|\mathcal F|=m\leq|\mathcal C|$.
Choose a fixed $d_0$ large enough for \cref{lem:quant-source},
\eqref{eq:direct-d-in-C} and all the estimates in this proposition. The proposition holds for every $N=N(d)$ with $d\geq d_0$.
\end{proof}

\begin{proof}[Proof of \cref{thm:quantitative}]
Let $\alpha_0$ be the constant of \cref{prop:quantitative-bound},
and consider an algorithm with a \kl{polynomial query bound}
$Q(\cdot)$ that achieves factor $\alpha_0\log\log|\mathcal C|$ with
success probability $p>0$. Choose $b\geq1$ with $Q(k)\leq k^b$ for
large $k$. The instances of \cref{prop:quantitative-bound} with
$|\mathcal C|=N$ have $|\mathcal F|\leq N$, so for large $N$ the
algorithm makes at most $(2N)^b\leq N^{2b}$ queries on them. By the
proposition, its success probability on one of these instances is
at most
\[
 \bigl(N^{2b}+N\bigr)
 \exp\!\left(-\frac19\left(\frac{\log N}{\log\log N}\right)^{2}\right)
 \leq\exp\!\left((2b+1)\log N
 -\frac19\left(\frac{\log N}{\log\log N}\right)^{2}\right),
\]
which tends to $0$ as $N\to\infty$ and is therefore below $p$ for large
$N$. This is a contradiction. Since
$\alpha(k)=o(\log\log k)$ implies $\alpha(k)\leq\alpha_0\log\log k$ for
large $k$, this gives the consequence for factor
$o(\log\log|\mathcal C|)$ stated after \cref{thm:quantitative}.
\end{proof}

\begin{remark}[Quasi-polynomial query bound]\label{rem:quasi-polynomial}
The proof above can be tightened. Fix $k\geq1$, and consider an algorithm that
achieves factor $\alpha_0\log\log|\mathcal C|$ with success
probability at least $|\mathcal C|^{-k}$. Let $Q$ be the maximum number
of queries of the algorithm over all runs on instances with
$|\mathcal C|=N$, for $N$ as in \cref{prop:quantitative-bound}. The
proposition gives
\[
 Q\geq N^{-k}\exp\!\left(\frac19\left(\frac{\log N}{\log\log N}\right)^{2}\right)-N
 =N^{\frac{\log N}{9(\log\log N)^2}-k}-N
 =N^{\Omega(\log N/(\log\log N)^2)},
\]
so the algorithm would require quasi-polynomially many queries.
\end{remark}
\begin{remark}[Positive distances]\label{rem:positive-distances}
The instances above place the $dR$ clients of each public block at one
point. Adding $1/d$ to the distance between any two distinct elements
of $\mathcal C\cup\mathcal V$ preserves the \kl{triangle inequality} and makes
all these distances positive. The resulting instances are the replicated
instances of the input with connection costs $\dist(u,v)+1/d$.
This shift raises the cost of every input assignment by
$|\mathcal C_0|/d=m$, so the shifted input has $\AbstractInputOpt\leq5m$ and
$\AbstractSymOpt\geq m\log d/2$, and the proof goes through with a smaller
constant $\alpha_0$.
\end{remark}

\begin{remark}[Expected guarantees]\label{rem:expected}
The lower bound also applies when only the expected number of queries
is polynomially bounded. Consider an algorithm that yields an approximation factor of $\alpha$ in expectation and that needs at most $Q(|\mathcal C|+|\mathcal F|)$ queries in expectation, for a polynomial $Q$. Stop
the algorithm after $4Q(|\mathcal C|+|\mathcal F|)$ queries and return an arbitrary assignment
if it has not finished. Truncation occurs with probability at most
$1/4$, and the original output exceeds $2\alpha\,\OPT$ with
probability at most $1/2$. The truncated algorithm therefore has a
polynomial query bound and returns a factor-$2\alpha$ assignment
with probability at least $1/4$. By \cref{thm:quantitative}, this cannot happen when $2\alpha\leq\alpha_0\log\log|\mathcal C|$ for large
$|\mathcal C|$, in particular when $\alpha=o(\log\log|\mathcal C|)$.

\end{remark}

\subsubsection{The configuration-LP integrality gap}
\label{sec:sfl-config-gap}

We can also use the \kl{lifted instance} in order to establish a lower bound on the integrality gap of the configuration LP~\eqref{eq:abstract-config-lp}.

\begin{corollary}[Configuration-LP gap for matroid facility costs]
\label{thm:sfl-config-gap}
The configuration LP for SFL has integrality gap
$\Omega(\log\log|\mathcal C|)$, even on instances with a common
\kl[rank function]{matroid rank} facility cost function. That is, there
are such instances with
\[
 \frac{\OPT}{\operatorname{ConfLP}}=\Omega(\log\log|\mathcal C|).
\]
\end{corollary}
\begin{proof}
Take the replicated instances of \cref{sec:replication} for the \kl{lifted instance} $(\widetilde {\mathcal U},\mathcal V,\dist,\SparseCost)$ with $R=1$.
The reference instance keeps the lifted clients and the metric, and
its facility cost function $\AbstractRefCost$ is a matroid rank function by \itemref{prop:abstract-instances}{item:abstract-matroid}.
By~\eqref{eq:sfl-input-gap}, Lemma~\itemref*{lem:abstract-optima}{item:abstract-ref-opt}
and Proposition~\itemref*{prop:abstract-config-lp}{item:abstract-lp} give
\begin{equation}\label{eq:sfl-config-abstract}
 \operatorname{ConfLP}_{\mathrm{ref}}\leq\AbstractInputOpt\leq4m,
 \qquad
 \AbstractRefOpt=\AbstractSymOpt
 \geq\frac{m\log d}{2}.
\end{equation}
Thus the gap is at least $(\log d)/8$. The instance has
$|\mathcal C|=dm=d^{2d+7}$ clients by \eqref{eq:quant-parameters}, so $\log\log|\mathcal C|=O(\log d)$.
\end{proof}

\section*{AI methodology statement}
ChatGPT~6 Astra (OpenAI) found a complete but hard-to-read proof, in
early September 2026, that no constant-factor approximation is possible
in the value-oracle model. Over the following weeks, the authors
verified and simplified this argument in conversations with Claude
Opus~5.5 (Anthropic) and ChatGPT~6 Astra. The improvement to the tight
$\Omega(\log\log|\mathcal C|)$ lower bound of \cref{thm:quantitative}
came out of a collaboration in these conversations: the authors pointed
out constructions that were suboptimal, and the models then found the
improved constructions. The authors used these tools to
draft and edit parts of the manuscript. The authors verified the
correctness and originality of all content, including all proofs and
references, and take full responsibility for it.

\begingroup
\small
\bibliographystyle{alpha}
\newcommand{\etalchar}[1]{$^{#1}$}

\endgroup

\clearpage
\appendix
\crefalias{section}{appendix}
\crefalias{subsection}{subappendix}
\counterwithin{theorem}{section}
\counterwithin{lemma}{section}
\counterwithin{corollary}{section}
\counterwithin{proposition}{section}
\counterwithin{fact}{section}
\counterwithin{definition}{section}
\counterwithin{remark}{section}
\counterwithin{question}{section}

\section{Computational hardness}\label{app:computational}

In this appendix, we convert the oracle lower bound into computational
hardness when the facility cost function is given as a circuit,
using the technique of Dobzinski and
Vondr\'ak~\cite{DobzinskiVondrak2012}.
We remark that the main result of~\cite{DobzinskiVondrak2012} is stated in terms of submodular maximization (i.e., the problem of finding a single set maximizing a given submodular function). As submodular facility location does not fall into this category, this result does not apply in a black box fashion.

The main goal of this section is to establish the following theorem.

\begin{restatable}{theorem}{ThmComputationalSymmetryGap}[Computational symmetry-gap theorem]\label{thm:computational-symmetry-gap}
	Fix a symmetric assignment instance
	$(\mathcal C_0,\mathcal F,\dist,\AbstractInputCost,\AbstractGroup)$
	whose facility costs and finite connection costs are rational and
	whose orbits all have size $q\in\mathbb Z_{\geq2}$. Let
	$\alpha\in\mathbb R_{\geq1}$ and $\nu\in\mathbb Q_{>0}$ satisfy
	\[
	\alpha(\AbstractInputOpt+\nu)<\AbstractSymOpt.
	\]
	Unless $\mathrm{NP}=\mathrm{RP}$, there is no randomized polynomial-time
	algorithm $\mathcal{A}$ with the following properties:
	\begin{enumerate}
		\item On any instance of submodular facility location, $\mathcal{A}$ terminates in polynomial time.
		\item On instances that arise as replications of $(\mathcal C_0,\mathcal F,\dist,\AbstractInputCost,\AbstractGroup)$ with matroid-rank facility
		costs represented by polynomial-size circuits, $\mathcal{A}$ achieves an approximation guarantee of $\alpha$ with constant success probability.
	\end{enumerate}
\end{restatable}

We point out that \cref{thm:computational-symmetry-gap} implies that unless $\mathrm{NP}=\mathrm{RP}$, SFL does not admit an $o(\log\log|\mathcal F|)$-approximation.
\begin{corollary}[Computational hardness for SFL]\label{cor:sfl-computational}
	Unless $\mathrm{NP}=\mathrm{RP}$, no randomized polynomial-time
	algorithm with constant success probability achieves approximation factor
	$o(\log\log|\mathcal F|)$ for SFL, even with a
	common matroid-rank facility cost function given by a polynomial-size
	evaluation circuit. In particular, SFL admits no constant-factor
	approximation in this representation.
\end{corollary}

\begin{proof}
	Take the cyclic input of \cref{sec:sfl-abstract-application}.
	It has $m=d^{2d+6}$ facilities, orbits of size $d$, and, by
	\eqref{eq:sfl-input-gap},
	\[
	\AbstractInputOpt\leq4m,
	\qquad
	\AbstractSymOpt\geq\frac{m\log d}{2}.
	\]
	With $\nu=m$, \cref{thm:computational-symmetry-gap} excludes every
	factor $\alpha<\log d/10$. 
Suppose towards a contraction that there were a randomized polynomial-time algorithm for SFL that, with constant success probability, achieves an approximation guarantee of $\beta(|\mathcal F|)$ for some function $\beta(k)\in o(\log\log k)$. Since
	$\log\log m=\Theta(\log d)$, we have $\beta(m)<\log d/10$ for
	some sufficiently large $d$ (which is also large enough to meet the guarantees of \cref{lem:quant-source}). Fix this $d$ and consider the corresponding cyclic input instance. Replication with factor $R$ leaves the $m$
	facilities unchanged, so the algorithm's
	guarantee is $\beta(m)$. This contradicts
	\cref{thm:computational-symmetry-gap}.
\end{proof}
The proof also excludes a sufficiently small positive constant times
$\log\log |\mathcal F|$ for sufficiently large
instances. Expected-cost guarantees are excluded as well, with a change
in the constants, by Markov's inequality as in \cref{rem:expected}.

We remark that in order to get a lower bound of $\Omega(\log\log |\mathcal{C}|)$, we would need the entire instance to grow, rather than just the replication factor. However, while we show that arbitrarily large instances with a high symmetry gap exist, the proof is not fully constructive (e.g., \cref{lem:quant-source}) and it is not immediately clear whether evaluation oracles for the involved functions can be implemented in polynomial time.

The remainder of this appendix is dedicated to the proof of \cref{thm:computational-symmetry-gap}.
In \cref{sec:RP,sec:ValiantVazirani}, we give a brief introduction to the complexity class $\mathrm{RP}$ and discuss a seminal result by Valiant and Vazirani~\cite{ValiantVazirani1986}, which shows, in particular, that any problem in $\mathrm{NP}$ admits a randomized reduction to the relaxation of SAT where we only need to correctly distinguish unsatisfiable formulas and those with a unique satisfying assignment. 

Following \cite{DobzinskiVondrak2012}, this motivates the following strategy to prove \cref{thm:computational-symmetry-gap}: Our goal is to convert the algorithm $\mathcal A$ for submodular facility location promised by \cref{thm:computational-symmetry-gap} into a randomized algorithm for the relaxed version of SAT discussed above as follows:
Given a CNF formula $\varphi$ on $m$ variables, we will choose a sufficiently large replication factor $R(m)$ and encode candidate solutions $z\in\{0,1\}^m$ to $\varphi$ via partitions $\pi_z$ of the replicated client set.
Using results from coding theory, we will show that this encoding can be done in such a way that from a given query $C$, we can recover a polynomially-sized list $L\subseteq \{0,1\}^m$ containing all $z\in\{0,1\}^m$ with the property that $C$ is not \kl[abstract well balanced]{well-balanced} for $\pi_z$.
This allows us to construct a polynomial-time algorithm, which, given a client set $C$, computes $\AbstractRefCost(C)$, if $\varphi$ is unsatisfiable, and $g_{\pi_z}(C)$, if $z$ is the unique satisfying assignment for $\varphi$; the behavior if $\varphi$ has multiple satisfying assignments is left unspecified.
To this end, we simply compute the list of all $z$ such that $C$ is unbalanced for $z$, and test whether any of them satisfies $\varphi$. If yes, we output $g_{\pi_z}(C)$, otherwise, we output $\AbstractRefCost(C)$. This is certainly correct if $\varphi$ is unsatisfiable, or if $\varphi$ has a unique satisfying assignment, which happens to be contained in $L$. On the other hand, if the unique satisfying assignment of $\varphi$ does not appear in $L$, then $C$ is \kl[abstract well balanced]{well-balanced} for $\pi_z$ and $\AbstractRefCost$ and $g_{\pi_z}$ agree on $C$.

Using standard arguments, this algorithm can be transformed into a circuit of polynomial size. We now run the algorithm $\mathcal{A}$ on the replicated instance with facility cost function given by the circuit. Then, we compute a polynomial-size list of all $z\in\{0,1\}^m$ such that at least one cluster of the solution is not \kl[abstract well balanced]{well-balanced} for $\pi_z$, and check whether any of these satisfies $\varphi$. If $\varphi$ is unsatisfiable, we clearly won't find any solution. However, if $\varphi$ has a unique satisfying assignment $z$, the circuit implements $g_{\pi_z}$, and with constant probability, $\mathcal A$ returns an $\alpha$-approximate solution. At least one of the clusters must be unbalanced for $\pi_z$; otherwise, the solution would cost at least as much as under $\AbstractRefCost$, which is more than $\alpha$ times the hidden optimum. Hence, the list of candidate solutions we produce contains $z$.

In \cref{sec:coding_theory}, we introduce the tools from coding theory that we need in order to encode candidate solutions $z$ via partitions of the client set and efficiently recover $z$ from a query $C$ that is unbalanced for $\pi_z$.
In \cref{sec:construction_circuits}, we formalize the proof sketched above on how to obtain the desired circuit encoding the facility cost function.
Finally, \cref{sec:proof_computation_symmetry_gap} proves \cref{thm:computational-symmetry-gap}.

\subsection{The complexity class $\mathrm{RP}$}\label{sec:RP}
\begin{definition}[Complexity class $\mathrm{RP}$]
The complexity class $\mathrm{RP}$ (randomized polynomial time) contains all decision problems for which there exists a randomized Turing machine with the following properties:
\begin{enumerate}
	\item It terminates on every input instance after a polynomial number of steps and outputs either `yes' or `no'.
	\item If the input is a no-instance, then the output is always `no'.
	\item If the input is a yes-instance, then the output is `yes' with probability at least $\frac{1}{2}$.
\end{enumerate}
\end{definition}
It is known that $\mathrm{RP}\subseteq \mathrm{NP}$.
\begin{definition}[randomized reduction]
A randomized reduction from decision problem $\mathcal A$ to decision problem $\mathcal B$ is a randomized Turing machine with the following properties:
\begin{enumerate}
	\item Given an instance $A$ of $\mathcal A$, it computes, in polynomial time, an instance $f(A)$ of $\mathcal B$.
	\item If $A$ is a no-instance of $\mathcal A$, then $f(A)$ is a no-instance of $\mathcal B$.
	\item If $A$ is a yes-instance of $\mathcal A$, then with probability at least $\frac{1}{2}$, $f(A)$ is a yes-instance of $\mathcal B$.
\end{enumerate}
\end{definition}
Note that if there is a randomized reduction from decision problem $\mathcal A$ to decision problem $\mathcal B$ and $\mathcal B$ is in $\mathrm{RP}$, then so is $\mathcal A$. Indeed, executing the reduction twice, executing the algorithm for $\mathcal B$ twice on each of the outputs, and returning yes if at least one of the four executions returns yes decreases the probability of failure to $(1-\frac{1}{4})^4<0.32$.

Valiant and Vazirani~\cite{ValiantVazirani1986} have shown that any $\mathrm{NP}$-complete decision problem $\mathcal A$, which asks whether or not a solution exists, admits a randomized reduction to its relaxation where only no-instances (no solution) and yes-instances with a unique solution need to be distinguished correctly. In particular, they remark that the existence of a randomized polynomial-time algorithm for SAT that is allowed to give arbitrary answers on instances with at least two solutions would imply $\mathrm{RP}=\mathrm{NP}$. This result is an immediate consequence of \cite{ValiantVazirani1986}, but differs from the statements formally proven there in the sense that rather than allowing arbitrary answers, they allow the answers to be specified by an arbitrary (but fixed) Boolean predicate. For completeness, we explain how to derive the statement we need from the results in~\cite{ValiantVazirani1986}. 

\subsection{Satisfiability and unique solutions}\label{sec:ValiantVazirani}
A Boolean variable $x$ can take values $0$ (false) or $1$ (true). We refer to $x$ and the negation $\overline{x}$ of $x$, which is true if and only if $x$ is false, as the \emph{literals} corresponding to $x$. A \emph{clause} (of size $k$) is formed by the disjunction $\ell_1\vee\ell_2\vee\dots\vee\ell_k$ of $k$ literals, it is true if and only if one of the literals evaluates to true. 
A \emph{Boolean formula} in \emph{conjuctive normal form (CNF)} consists of the conjuction $C_1\wedge C_2\wedge \dots\wedge C_m$ of a finite number of clauses; it evaluates to true if and only if each clause does. The \emph{size} of a CNF formula is the total size of its clauses.

\begin{definition}[Satisfiability (SAT)]
The decision problem satisfiability (SAT) takes as input a boolean formula in CNF, and asks whether this formula admits a satisfying variable assignment.
\end{definition}
\begin{theorem}[\cite{ValiantVazirani1986}]\label{thm:ValiantVazirani}
Let $S\subseteq \{0,1\}^n$ be nonempty. Let $w_1,\dots,w_n$ be vectors chosen independently and uniformly at random from $\{0,1\}^n$. For $i\in [n]$, let $H_i\coloneqq \{x\in \{0,1\}^n\colon \sum_{j=1}^n w_{i,j}\cdot x_j\equiv 0 \mod 2\}$ define $S_i\coloneqq S\cap H_1\cap \dots\cap H_i$.
Then $\Pr[\exists i : |S_i|=1]\ge \frac{1}{4}$. 
\end{theorem} 
We remark that $H_i=\{x\in \{0,1\}^n\colon\bigoplus_{j : w_{i,j=1}} x_j =0\}$, where $\bigoplus$ denotes the exclusive or. $H_i$ can be expressed as the set of satisfying assignments of a CNF formula $\Phi_{w_i}$ of linear size.

\begin{proposition}[\cite{ValiantVazirani1986}]\label{prop:ValiantVazirani}
For $n\in\mathbb{Z}_{\ge 1}$, there exists a CNF formula $\Phi_n(x_1,\dots,x_n,y_1,\dots,y_{m})$ with $m=\max\{n-3,0\}$ of size $O(n)$ such that for any assignment $(x_i)_{i=1}^n\in\{0,1\}^n$, there exist $(y_i)_{i=1}^m\in \{0,1\}^m$ with $\Phi_n(x_1,\dots,x_n,y_1,\dots,y_m)$ if and only if $\bigoplus_{i=1}^n x_i=0$.
\end{proposition}
\begin{proof}
	By induction on $n$.
For $n=1,2,3$, we have $m=0$ and define $\Phi_1(x_1)=\overline{x}$, $\Phi_2(x_1,x_2)=(x_1\vee \overline{x_2})\wedge (\overline{x_1}\vee x_2)$, and 
$\Phi_3(x_1,x_2,x_3)=(x_1\vee \overline{x_2}\vee x_3)\wedge (\overline{x_1}\vee x_2\vee x_3)\wedge (x_1\vee x_2\vee \overline{x_3})\wedge (\overline{x_1}\vee \overline{x_2}\vee \overline{x_3})$.
Next, let $n\ge 4$. We have $\bigoplus_{i=1}^n x_i=0$ if and only if there is $y_1$ such that $\bigoplus_{i=1}^{n-2} x_i\oplus y_1=0$ and $y_1\oplus x_{n-1}\oplus x_n = 0$. Let $\Phi_n(x_1,\dots,x_n,y_1,\dots,y_{n-3})\coloneqq \Phi_{n-1}(x_1,\dots,x_{n-2},y_{1},y_2,\dots,y_{n-3})\wedge \Phi_3(x_{n-1},x_{n},y_1)$.
\end{proof} 
\cref{thm:ValiantVazirani,prop:ValiantVazirani} imply the following result:
\begin{corollary}[\cite{ValiantVazirani1986}]\label{cor:new_formulas_unique_assignment}
There exists a polynomial-time randomized algorithm that, given a CNF formula $\Phi$ on $n$ variables, computes $n$ CNF formulas $(\Phi_i)_{i=1}^n$ with the following properties:
\begin{enumerate}
	\item If $\Phi$ is not satisfiable, then none of the $\Phi_i$ is.
	\item If $\Phi$ is satisfiable, then with probability at least $\frac{1}{4}$, there is $i\in [n]$ such that $\Phi_i$ has a unique satisfying assignment.
\end{enumerate}
\end{corollary}
\begin{corollary}[\cite{ValiantVazirani1986}]\label{cor:decide_unique_SAT_NP_RP}
Let $p\in (0,1]$ be a constant. Suppose there were a randomized polynomial-time algorithm for SAT with the following properties:
\begin{enumerate}
	\item On a no-instance, the algorithm outputs `no'.
	\item On a yes-instance with a unique satisfying assignment, the algorithm outputs `yes' with probability at least $p$.
\end{enumerate}
Then $\mathrm{NP}=\mathrm{RP}$.
\end{corollary}
\begin{proof}
Suppose there were an algorithm $T$ as stated in the corollary. Consider the following randomized polynomial-time algorithm for SAT. Given a CNF formula $\Phi$, we first run the algorithm from \cref{cor:new_formulas_unique_assignment} five times, giving formulas $(\Phi^{(j)}_i)_{i=1}^n$ for $j\in [5]$. Then, we run $T$ $L\coloneqq\lceil -\frac{2}{\log_2(1-p)}\rceil$ times on each of these $5n$ formulas. If we always get `no' as an answer, we output `no', otherwise, we output `yes'.
If $\Phi$ is not satisfiable, then none of the instances we run $T$ on is, and we always get `no' for an answer.
On the other hand, if $\Phi$ is satisfiable, then with probability $1-\left(\frac{3}{4}\right)^5>\frac{3}{4}$, at least one of the $5n$ instances we construct, say $\Phi^{(j)}_i$, has a unique solution. The executions of $T$ on this instance all return `no' with probability at most $(1-p)^L\le\frac{1}{4}$, i.e., at least one of them returns yes with probability $\frac{3}{4}$. Our total success probability is hence at least $\frac{3}{4}\cdot\frac{3}{4}=\frac{9}{16}>\frac{1}{2}$. This implies that SAT is in $\mathrm{RP}$. As SAT is further $\mathrm{NP}$-complete and $\mathrm{RP}\subseteq \mathrm{NP}$, this yields $\mathrm{NP}=\mathrm{RP}$.
\end{proof}

\subsection{Tools from coding theory}\label{sec:coding_theory}
\begin{definition}[linear code]
	Let $Q\in\mathbb Z_{>0}$ be a prime power and let $\mathbb F_Q$ be the finite field of order $Q$.
	A \intro[linear code]{error-correcting linear code} over a $Q$-ary alphabet of \intro{block length} $n$ is a linear subspace $\intro*\Code\subseteq \mathbb{F}_Q ^n$. Its dimension $k$ is called the information content of the code.
\end{definition}
We call a \kl[linear code]{code} \intro{polynomially constructible} if $n$ is polynomially bounded in $k$ and there exist algorithms $\intro*\Enc$ and $\intro*\Dec$ with the following properties:
\begin{enumerate}
	\item $\Enc$ computes, in time polynomial in $k$, an injective function from $\mathbb{F}^k_Q$ to $\Code$.
	\item Given $y\in\mathbb{F}_Q^n$, $\Dec$ outputs, in time polynomial in $k$, either `$y\notin \Code$' or $z\in\mathbb{F}_Q^k$ with $y=\Enc(z)$.
\end{enumerate} We call $\Enc$ the \intro{encoder} and $\Dec$ the \intro{decoder}.
\begin{definition}
	Given a \kl{linear code} $\Code\subseteq \mathbb{F}_Q ^n$, a \intro{list decoder} for $e \in \mathbb{Z}_{\ge 0}$ errors is an algorithm $\intro*\ListDec$, which, given $y\in \mathbb{F}_Q ^n$, computes $\{x\in \Code\colon |\{i\colon x_i\ne y_i\}|\le e\}$.
\end{definition}
The following theorem follows from Theorem 8 and Corollary 2 in~\cite{GuruswamiSudan2000}, together with the explicit construction of the used concatenation of Reed-Solomon and Hadamard codes.
\begin{theorem}[\cite{GuruswamiSudan2000}]\label{theorem:good_codes}
	Let $Q\in\mathbb Z_{>0}$ be a prime power and let $\varepsilon \in (0,1)$. For every $m\in\mathbb Z_{>0}$, there exists a \kl{polynomially constructible} error-correcting \kl[linear code]{linear code} $\Code\subseteq \mathbb{F}_Q ^n$ of dimension $m$ and \kl{block length} $n = O(m^2/\varepsilon^4)$ that admits a polynomial-time \kl{list decoder} for $(1-1/Q)(1-\varepsilon)\cdot n$ errors.
\end{theorem}
The lemma below is essentially contained in~\cite[Section~4.2]{DobzinskiVondrak2012}.
 We first restate the relevant result from there in our notation.
\begin{lemma}[{\cite[Lemma~4.15]{DobzinskiVondrak2012}}]\label{lemma:dv_rounding}
	There is an algorithm that, given $Q,n\in\mathbb Z_{>0}$ and $S\subseteq [n]\times [Q]$, computes in time $\mathrm{poly}(n,Q)$ a list of $nQ$ words in $[Q]^n$ with the following property: for every $\gamma>0$ and every $y\in [Q]^n$ with
	\[|\{i\in [n]\colon (i,y_i)\in S\}|\ge \frac{1}{Q}|S|+\gamma n,\]
	at least one word $w$ on the list satisfies $|\{i\in[n]\colon w_i\ne y_i\}|\le \left(1-\frac{1+\gamma}{Q}\right)n$.
\end{lemma}
\Cref{lemma:recover_paritition_from_unbalanced_query} below is an adaptation of~\cite[Claim~4.16]{DobzinskiVondrak2012}.
Claim~4.16 needs a list-decodable code over the alphabet $\mathbb Z_q$. While \cite{DobzinskiVondrak2012} do not restrict $q$ to prime powers, the code construction they cite~\cite{GuruswamiRudra2008} is stated over finite fields, i.e., prime-power alphabets, and so is \cref{theorem:good_codes}. To be safe, we use a code over $\mathbb F_Q$ for a sufficiently large power of two $Q$ and reduce modulo $q$ (Line~\ref{line:recover-Sprime}). Moreover, we allow the imbalance in either direction, which the complements $S_t^c$ handle (Line~\ref{line:recover-family}).
\begin{lemma}\label{lemma:recover_paritition_from_unbalanced_query}
	Let $\delta\in (0,1)$ and $q\in\mathbb{Z}_{\ge 2}$ be constants. There exist polynomial-time algorithms $\mathcal E$ and $\mathcal L$ such that for every $m\in\mathbb{Z}_{>0}$, there is $n=\mathrm{poly}(m)$ with the following properties:
	\begin{enumerate}
		\item Given $z\in \{0,1\}^m$, $\mathcal{E}$ computes a function $f_z\colon [n]\times \mathbb Z_q\rightarrow\mathbb Z_q$ such that $|f_z^{-1}(\{t\})|=n$ for every $t\in \mathbb Z_q$.
		\item Given $m$ and $C\subseteq [n]\times \mathbb Z_q$, $\mathcal L$ outputs a list of $\mathrm{poly}(m)$ many strings in $\{0,1\}^m$. 
		This list is guaranteed to contain $z\in\{0,1\}^m$ if there exists $t\in\mathbb Z_q$ such that $\left|\frac{|C|}{q}-|f_z^{-1}(\{t\})\cap C|\right|\ge\delta n$.
	\end{enumerate}
\end{lemma}
\begin{proof}
	Let $\gamma\coloneqq \min\{\frac{\delta}{2},\frac{1}{4}\}$ and let $Q =  2^{\lceil \log_2(\frac{q}{\gamma})\rceil}$ be a power of two not smaller than $q / \gamma$. Fix a bijection $\intro*\SymbolMap\colon\mathbb{F}_Q \to [Q]$.
	For $m\in \mathbb{Z}_{>0}$, apply \cref{theorem:good_codes} with parameters $Q$ and $\varepsilon\coloneqq\frac{\gamma}{Q}$ in order to obtain an error-correcting linear code of dimension $m$.
	Let $n$, $\Code$, $\Dec$, $\Enc$, and $\ListDec$ be the \kl{block length}, \kl[linear code]{code}, \kl{decoder}, \kl{encoder}, and \kl{list decoder} given by \cref{theorem:good_codes}.
	Then $n= O(m^2/\varepsilon^4)$ and each of the mentioned algorithms runs in polynomial-time. 
	The error bound for the list decoder amounts to $(1-1/Q)\cdot (1-\gamma/Q)n\ge (1-1/Q-\gamma/Q)\cdot n=n-(1+\gamma)/Q\cdot n$, meaning that for every $w\in\mathbb F_Q^n$, $\ListDec(w)$ contains every codeword that agrees with $w$ in at least $(1+\gamma)/Q\cdot n$ positions. 
	Given $z\in\{0,1\}^m$, the algorithm $\mathcal{E}$ proceeds as follows:
	\begin{itemize}
		\item Compute $y\coloneqq\Enc(z)\in\mathbb{F}_Q^n$, interpreting $z$ as a vector in $\mathbb{F}_Q^m$.
		\item For $i \in [n]$ and $t \in \mathbb{Z}_q$ define $f_z(i,t) \coloneqq (t + \SymbolMap(y_i)) \bmod q$.
	\end{itemize}
	The algorithm $\mathcal{E}$ runs in polynomial time because $q$ and $Q$ are constant and $n$ is polynomially bounded in $m$. For $t\in \mathbb{F}_q$, $f_z^{-1}(\{t\})=\{(i,(t-\SymbolMap(y_i)) \bmod q)\colon i\in [n]\}$ is a set of size $n$.
	
	Next, we describe the algorithm $\mathcal{L}$, see \cref{alg:recover}.

	\begin{algorithm}[H]
		\caption{The algorithm $\mathcal{L}$}\label{alg:recover}
		\KwIn{$m$ and $C\subseteq [n]\times \mathbb{Z}_q$}
		\KwOut{a list of strings in $\{0,1\}^m$}
		\For{$t\in \mathbb{Z}_q$}{
			$S_t\gets \{(i,j)\in [n]\times \mathbb{Z}_q\colon \big(i,(t-j) \bmod q\big)\in C\}$\;\label{line:recover-St}
			$S_t^c\gets ([n]\times \mathbb{Z}_q)\setminus S_t$\;
		}
		$\mathcal S\gets \{S_t\colon t\in \mathbb{Z}_q\}\cup\{S_t^c\colon t\in \mathbb{Z}_q\}$\;\label{line:recover-family}
		$\mathrm{LIST}\gets\emptyset$\;
		\For{$S\in \mathcal S$}{
			$S'\gets \{(i,t)\in [n]\times [Q]\colon (i,t \bmod q)\in S\}$\;\label{line:recover-Sprime}
			$W_S\gets$ the list of words in $[Q]^n$ computed by the algorithm of \cref{lemma:dv_rounding} on input $S'$\;\label{line:recover-round}
			\For{$w\in W_S$}{
				$\mathrm{LIST}\gets \mathrm{LIST}\cup \ListDec\big((\SymbolMap^{-1}(w_i))_{i\in[n]}\big)$\;\label{line:recover-listdec}
			}
		}
		$\mathrm{OUTPUT}\gets\emptyset$\;
		\For{$y\in\mathrm{LIST}$}{
			\If{$\Dec$ outputs $z\in\{0,1\}^m$ on input $y$}{
				$\mathrm{OUTPUT}\gets \mathrm{OUTPUT}\cup\{z\}$\;
			}
		}
		\Return{$\mathrm{OUTPUT}$}\;
	\end{algorithm}
	The algorithm runs in polynomial-time and returns a list of strings in $\{0,1\}^m$ of polynomial size.
	
	Suppose that for some $z\in \{0,1\}^m$, there is $t\in \mathbb{Z}_q$ such that $\left||f^{-1}_z(t)\cap C|-\frac{|C|}{q}\right|\ge \delta n$.
	Let $y\coloneqq \Enc(z)$ and $\bar y_i\coloneqq \SymbolMap(y_i)\bmod q$ for $i\in[n]$.
	We have $f^{-1}_z(\{t\})=\{(i,t-\SymbolMap(y_i)) : i\in [n]\}$, and, as such, we have 
	\[|f^{-1}_z(\{t\})\cap C|=|\{i\in [n] : (i,t-\SymbolMap(y_i))\in C\}|=|\{i\in [n] : (i,\bar y_i) \in S_t\}|.\]
	Moreover, as $(i,j \bmod q)\mapsto (i,(t-j) \bmod q)$ induces a bijection between $C$ and $S_t$, we have $|C|=|S_t|$. This allows us to conclude that
	\[\left||\{i\in [n] : (i,\bar y_i)\in S_t\}|-\frac{|S_t|}{q}\right|=\left||f^{-1}_z(t)\cap C|-\frac{|C|}{q}\right|\ge \delta n.\]
	We have $|\{i\in[n] : (i,\bar y_i)\in S_t^c\}|=n-|\{i\in [n]\colon (i,\bar y_i)\in S_t\}|$ and $\frac{|S_t^c|}{q} = \frac{nq-|S_t|}{q}=n-\frac{|S_t|}{q}$, so exactly one set $S\in \{S_t,S_t^c\}$ satisfies
	\[|\{i\in [n]\colon (i,\bar y_i)\in S\}|-\frac{|S|}{q}\ge \delta n.\]
	Our next goal is to transfer this inequality to $S'$. To this end, we first observe that by definition of $S'$, we have $(i,\SymbolMap(y_i))\in S'$ if and only if $(i,\bar y_i)\in S$, so $|\{i\in [n]\colon (i,\bar y_i)\in S\}|=|\{i\in [n]\colon (i,\SymbolMap(y_i))\in S'\}|$. Next, we observe that $|S'|\le (\frac{Q}{q}+1)\cdot |S|$ because every $t\in \mathbb{Z}_q$ has at most $\frac{Q}{q}+1$ many preimages under the map $[Q]\rightarrow\mathbb{Z}_q,t\mapsto t\mod q$. Finally, we have $\frac{|S|}{Q}\le \frac{nq}{Q}\le \gamma n\le\frac{\delta}{2}n$.
	This gives
	\begin{align*}
		\phantom{=}|\{i\in [n]\colon (i,\SymbolMap(y_i))\in S'\}|-\frac{|S'|}{Q}
		&\ge |\{i\in [n]\colon (i,\bar y_i)\in S\}|-\frac{|S|}{Q}\cdot \left(\frac{Q}{q}+1\right)\\
		&=|\{i\in [n]\colon (i,\bar y_i)\in S\}|-\frac{|S|}{q}-\frac{|S|}{Q}\\
		&\ge \delta n-\frac{\delta}{2}n=\frac{\delta}{2}n \ge \gamma n.
	\end{align*}
	Hence, $(\SymbolMap(y_i))_{i\in[n]}\in[Q]^n$ satisfies the hypothesis of \cref{lemma:dv_rounding} for $S'$, so some $w\in W_S$ in Line~\ref{line:recover-round} agrees with it in at least $(1+\gamma)/Q\cdot n$ positions. Then $(\SymbolMap^{-1}(w_i))_{i\in[n]}$ agrees with $y$ in the same positions, so it differs from $y$ in at most $\left(1-\frac{1+\gamma}{Q}\right)n\le \left(1-\frac1Q\right)\left(1-\frac{\gamma}{Q}\right)n$ positions, which is within the list-decoding radius of \cref{theorem:good_codes} for $\varepsilon=\frac{\gamma}{Q}$. Hence, $\ListDec$ applied to it contains $y$, and applying $\Dec$ to $y$ adds $z$ to the output. 
\end{proof}
\subsection{Construction of circuits}\label{sec:construction_circuits}
\begin{lemma}\label{lemma:compute_circuit_for_functions}
	Let $U$ be a finite set and let $R\colon \mathbb{N}\rightarrow\mathbb{N}$ be polynomially bounded and computable in polynomial time. Suppose we are given polynomial-time algorithms $H_1$, $H_2$ and $L$ with the following properties:
	\begin{itemize}
		\item Given $m$ (in unary) and $C\subseteq [R(m)]\times U$, $H_1$ computes $H_1(m,C)\in \{0,\dots,R(m)\cdot |U|\}$.
		\item Given $z\in\{0,1\}^m$ and $C\subseteq [R(m)]\times U$, $H_2$ computes $H_2(z,C)\in \{0,\dots,R(m)\cdot |U|\}$.
		\item Given $m$ (in unary) and $C\subseteq [R(m)]\times U$, $L$ computes a set $L_m(C)\subseteq \{0,1\}^m$ of polynomial size such that $\{z\colon H_1(m,C)\ne H_2(z,C)\}\subseteq L_m(C)$.
	\end{itemize} 
	Then there exists a polynomial-time algorithm, that, given a Boolean CNF formula $\varphi$ in $m$ variables, constructs a circuit of size polynomial in $R(m)$ that computes a function $f\colon 2^{[R(m)]\times U}\rightarrow\{0,\dots, R(m)\cdot |U|\}$ that agrees with $H_1(m,\cdot)$, if $\varphi$ is unsatisfiable, and with $H_2(z,\cdot)$, if $\varphi$ is uniquely satisfiable with satisfying assignment $z$.
\end{lemma}
\begin{proof}
	It is well-known that a deterministic polynomial-time algorithm computing an integer valued function (more precisely, a corresponding Turing-machine), can be transformed, in polynomial-time, into a Boolean circuit of polynomial size computing said function, see, e.g.,~\cite[Chapter~6]{Arora_Barak_2009}.
	As such, it suffices to present a polynomial-time algorithm to compute a function $f$ with the described properties.
	Given $m$ and $C\subseteq [R(m)]\times U$, our algorithm proceeds as follows:
	\begin{enumerate}
		\item Apply $L$ to compute $L_m(C)$.
		\item For every $z\in L_m(C)$, check if $z$ is a satisfying assignment for $\varphi$. If yes, apply $H_2$ to output $H_2(z,C)$.
		\item If no satisfying assignment was found in the previous step, we apply $H_1$ to output $H_1(m,C)$.
	\end{enumerate}
	The algorithm runs in polynomial-time and, on query $C$, outputs $H_1(m,C)$ or $H_2(z,C)$ (for some $z\in \{0,1\}^m$). In either case, its output is an integer in $\{0,\dots,R(m)\cdot |U|\}$.
	
	Suppose that $\varphi$ is unsatisfiable. In that case, we cannot find a satisfying assignment in the second step, so we output $H_1(m,C)$ as required.
	
	Next, suppose that $\varphi$ has a unique satisfying assignment $z$. If $z\in L_m(C)$, then we output $H_2(z,C)$ as required because all assignments we might try before $z$ will not satisfy $\varphi$. On the other hand, if $z\notin L_m(C)$, then we won't find a satisfying assignment in the second step, and we output $H_1(m,C)$. By the guarantees of the algorithm $L$, we have $H_1(m,C)=H_2(z,C)$ in this case.
\end{proof}
\subsection{Proof of \cref{thm:computational-symmetry-gap}}\label{sec:proof_computation_symmetry_gap}

\ThmComputationalSymmetryGap*

\begin{proof}[Proof of \cref{thm:computational-symmetry-gap}]
	Suppose that an algorithm $\mathcal{A}$ as in the statement of the theorem exists.
	
	We use it to construct a randomized algorithm for SAT as in \cref{cor:decide_unique_SAT_NP_RP}; this will imply $\mathrm{NP}=\mathrm{RP}$.

We begin by defining the notation and parameters needed for the replication of the given symmetric assignment instance.	
Put $\delta\coloneqq\nu/|\mathcal C_0|\in\mathbb Q_{>0}$.
Choose $D\in\mathbb Z_{>0}$ so that
$D\AbstractInputCost\colon2^{\mathcal C_0}\to\mathbb Z_{\geq0}$ and
$D\delta\in\mathbb Z$. Let $R$ be a positive multiple of $D$ (we fix $R=R(m)$ below), and take replication
factor $\AbstractCopies=R$ and penalty weight
$\AbstractPenaltyWeight=\lambda\coloneqq R\delta$. 
 We identify $\mathcal{C}_0$ with $I\times \mathbb{Z}_q$, where the $i$-th orbit $P_i$ corresponds to $\{i\}\times \mathbb{Z}_q$, and we identify $\mathcal C$ with $[R]\times I\times\mathbb{Z}_q$, where the public block $B(i)$ corresponds to $[R]\times\{i\}\times\mathbb{Z}_q$. 

 For a hidden partition $\pi\colon\mathcal C\to\mathcal C_0$, $\pi(B(i))=P_i$ for all $i\in I$, let $h_\pi\colon2^{\mathcal C}\to\mathbb Q_{\geq0}$ be the
hidden function $\AbstractHidCost$ from~\eqref{eq:abstract-hidden-function}, and let
$h_{\mathrm{ref}}=\AbstractRefCost$ be the reference function from~\eqref{eq:abstract-reference-function}.
Since $R$ is a multiple of $D$, both $R\AbstractInputCost$ and
$\lambda$ are integral. Hence \cref{prop:abstract-instances} shows
that $h_{\mathrm{ref}}$ and $h_{\pi}$ are matroid rank functions, with values in
$\{0,\ldots,|\mathcal C|\}$. \cref{lem:abstract-optima} shows that
\begin{equation}\label{eq:computational-optima}
	\AbstractHidOpt\leq R(\AbstractInputOpt+\nu),
	\qquad
	\AbstractRefOpt=R\AbstractSymOpt.
\end{equation}

Apply \cref{lemma:recover_paritition_from_unbalanced_query} with parameters $\delta$ and $q$ to obtain algorithms $\mathcal{E}$ and $\mathcal{L}$ as specified in the lemma, and let $n\colon \mathbb{Z}_{>0}\rightarrow\mathbb{Z}_{>0}$ be the corresponding polynomially bounded function. Let $R\colon \mathbb{Z}_{>0}\rightarrow\mathbb{Z}_{>0}$ be given by $R(m)=D\cdot n(m)$.
Given a CNF formula with $m$ variables, we will choose the replication factor to be $R(m)$.
For $z\in\{0,1\}^m$, let $f_z$ be the function computed by $\mathcal{E}$ on input $z$, and define the hidden partition $\pi_z\colon \mathcal{C}\cong [R(m)]\times I\times\mathbb Z_q\rightarrow \mathcal{C}_0\cong I\times\mathbb Z_q$ by $\pi_z(a,i,t)\coloneqq (i,f_z(a \mod n(m),t))$. Note that $\pi_z$ maps $B(i)$ to $P_i$ and moreover, the preimage of a base client $(i,t)\in \mathcal C_0$ contains exactly $|f_z^{-1}(t)|\cdot D=R(m)$ clients from $\mathcal C$.

Next, we explain how to, given a CNF formula $\varphi$, obtain a circuit of polynomial size that computes $h_{\mathrm{ref}}$, if $\varphi$ is unsatisfiable, and $h_{\pi_z}$, if $\varphi$ has a unique satisfying assignment $z$. Write $n\coloneqq n(m)$ and $R\coloneqq R(m)$.
We now describe the algorithms $H_1$, $H_2$ and $L$ needed for \cref{lemma:compute_circuit_for_functions}.
 \begin{description}
 	\item[$H_1$:] Given $m$ and $C\subseteq [R(m)]\times \mathcal{C}_0$, $H_1$ computes $h_{\mathrm{ref}}(C)$. Note that this can indeed be done in polynomial-time: The symmetric assignment instance $(\mathcal C_0,\mathcal F,\dist,\AbstractInputCost,\AbstractGroup)$ is fixed in \cref{thm:computational-symmetry-gap} and does not depend on $\varphi$, so $|\mathcal C_0|$ is a constant and as such, we can, in constant time, write down the complete value table of the function $\AbstractInputCost$. Moreover, $|I|$ is a constant, so the minimum in \eqref{eq:abstract-reference-function} can be evaluated in polynomial time.
 	\item[$H_2$:] Given $z\in \{0,1\}^m$ and $C\subseteq [R(m)]\times \mathcal{C}_0$, $H_2$ applies $\mathcal{E}$ to compute $f_z$, and hence $\pi_z$, and then computes $h_{\pi_z}(C)$, which can be done in polynomial time by evaluating \eqref{eq:abstract-hidden-function} as $|\mathcal C_0|$ is constant.
 	\item[$L$:] Given $m$ and $C\subseteq [R(m)]\times \mathcal{C}_0$, $L$ computes a polynomial-size set $L_m(C)\supseteq\{z\colon H_1(m,C)\ne H_2(z,C)\}$. By \cref{lem:abstract-agreement}, it suffices that $L_m(C)$ contains all $z\in \{0,1\}^m$ such that $C$ is not \kl[abstract well balanced]{well-balanced} for $\pi_z$ (with parameter $\AbstractPenaltyWeight=\delta R$). For $i\in I$, let $\widetilde{C}_i\coloneqq \{(a,j)\in [R(m)]\times\mathbb{Z}_q\colon (a,i,j)\in C\}$. Consider a base client $u=(i,t)\in\mathcal C_0$. We have
	\[\AbstractBlockAvg_u(C)=\frac{|C\cap B(i)|}{|P_i|}=\frac{|\widetilde{C}_i|}{q}\text{ and}\]
 	\[\AbstractCopyCount_u(C)=|C\cap \pi_z^{-1}(u)|=|\{(a,j)\in\widetilde{C}_i\colon  f_z(a\mod n,j)=t\}|.\] As such, $C$ is not \kl[abstract well balanced]{well-balanced} for $\pi_z$ if and only if there exists $(i,t)\in \mathcal C_0$ with
 	\begin{equation}
		\left| \frac{|\widetilde{C}_i|}{q}-|\{(a,j)\in\widetilde{C}_i : f_z(a\bmod n,j)=t\}|\right| > \delta R=\delta nD.
 \end{equation} 
 For $d\in [D]$, let $\widetilde{C}_{i,d}\coloneqq \{(r,j)\in[n]\times \mathbb{Z}_q\colon ((d-1)n+r,j)\in\widetilde{C}_i\}$. Note that $|\widetilde{C}_i|=\sum_{d\in [D]}|\widetilde{C}_{i,d}|$ and $|\{(a,j)\in\widetilde{C}_i : f_z(a\bmod n,j)=t\}|=\sum_{d\in [D]} |\widetilde{C}_{i,d}\cap f_z^{-1}(t)|$. Hence, if $C$ is not \kl[abstract well balanced]{well-balanced} for $\pi_z$, then there exists $d\in [D]$ such that 
 	\[\left| \frac{|\widetilde{C}_{i,d}|}{q}-|\widetilde{C}_{i,d}\cap f_z^{-1}(t)|\right|> \delta n,\] i.e., $\widetilde{C}_{i,d}\subseteq [n]\times \mathbb{Z}_q$ is $\delta n$-unbalanced for the partition induced by $f_z$. Hence, we can obtain a polynomially-sized list including all such vectors $z$ by applying the algorithm $\mathcal{L}$ to all of the sets of the form $\widetilde{C}_{i,d}$, which can be computed from $C$ in polynomial time.
 \end{description}
 We are now ready to describe our algorithm for SAT. Given a CNF formula $\varphi$ in $m$ variables, we apply \cref{lemma:compute_circuit_for_functions} with the algorithms $H_1$, $H_2$ and $L$ described above, to obtain a circuit $\Lambda$ of polynomial size that computes a function $f\colon 2^{\mathcal C}\rightarrow\{0,\dots,|\mathcal C|\}$ which agrees with $h_{\mathrm{ref}}$, if $\varphi$ is unsatisfiable, and with $h_{\pi_z}$, if $\varphi$ is uniquely satisfiable with satisfying assignment $z$. In either of these cases, the circuit computes a matroid rank function.
 
 Now, we apply the randomized algorithm $\mathcal A$ for submodular facility location to the instance, with facility cost function given by the circuit $\Lambda$. We stop $\mathcal A$ after polynomially many steps, as guaranteed for instances of submodular facility location, and output `no' if it has not terminated by then. This only matters if $\varphi$ has several satisfying assignments, since otherwise $\Lambda$ computes a matroid rank function. Let the output be clusters $(C_f)_{f\in \mathcal{F}}$. We apply the algorithm $L$ to each cluster $C_f$ and output `yes' if and only if some $z\in\bigcup_{f\in\mathcal F}L_m(C_f)$ satisfies $\varphi$.
 
 If $\varphi$ is unsatisfiable, no candidate satisfies $\varphi$, so we output `no'. If $\varphi$ has multiple satisfying assignments, there is nothing to prove. The interesting case appears when $\varphi$ has exactly one satisfying assignment $z$. In this case, $\Lambda$ encodes the matroid rank function $h_{\pi_z}$ and, hence, with constant probability, the solution $(C_f)_{f\in\mathcal{F}}$ returned by $\mathcal{A}$ is an $\alpha$-approximation of the optimum and by \eqref{eq:computational-optima}, we can bound its cost by
 \[
 \alpha\AbstractHidOpt
 \leq\alpha R(\AbstractInputOpt+\nu)
 <R\AbstractSymOpt=\AbstractRefOpt.
 \]  In particular, the solution is cheaper than the optimum with respect to $h_{\mathrm{ref}}$. But this means that for some cluster $C_f$, we must have $h_{\pi_z}(C_f)\ne h_{\mathrm{ref}}(C_f)$, i.e., $H_2(z,C_f)\ne H_1(m,C_f)$. By the guarantee of $L$, we have $z\in L_m(C_f)$, so we output `yes'.
 \end{proof}

\clearpage
\end{document}